\documentclass[10pt]{article}
\RequirePackage[OT1]{fontenc}
\usepackage{a4wide}
\usepackage{amsmath,amsthm,amssymb}
\usepackage{amsfonts,dsfont}
\usepackage{comment}
\usepackage[longnamesfirst,round]{natbib}
\usepackage{colordvi}
\usepackage{hyperref}
\usepackage{graphicx}
\usepackage{rotating}
\usepackage{array}
\newtheorem{theorem}{Theorem}[section]
\newtheorem{proposition}[theorem]{Proposition}
\newtheorem{corollary}[theorem]{Corollary}
\newtheorem{conj}[theorem]{Conjecture}
\newtheorem{lemma}[theorem]{Lemma}

\newtheorem{definition}[theorem]{Definition}
\newcommand{\R}{\mathbb{R}} 

\renewcommand{\P}{\mathbb{P}}
\newcommand{\E}{\mathbb{E}} 
\newcommand{\telque}{\;:\;}
\newcommand{\FDR}{\mbox{FDR}}

\newcommand{\mbe}{\mathbb{E}}
\newcommand{\mbr}{\mathbb{R}}
\newcommand{\mbp}{\mathbb{P}}

\newcommand{\mbn}{\mathbb{N}}
\newcommand{\bp}{\mathbf{p}}

\newcommand{\mbf}{\mathbf}

\newcommand{\wh}[1]{{\widehat{#1}}}
\newcommand{\ol}[1]{\overline{#1}}

\newcommand{\ind}[1]{{\mbf{1}\{#1\}}}

\newcommand{\prob}[1]{\mbp\brac{#1}}

\newcommand{\brac}[1]{\left[#1\right]}

\newcommand{\set}[1]{\left\{#1\right\}}

\newcommand{\eps}{\varepsilon}

\newcommand{\LSU}{\mbox{LSU}}
\newcommand{\LSD}{\mbox{LSD}}
\newcommand{\SUD}{\mbox{SUD}_{\lambda}}
\newcommand{\SUDm}{\mbox{SUD}_{\lambda_m}}
\newcommand{\LSUD}{\mbox{LSUD}_{\lambda}}

\newcommand{\SU}{\mbox{SU}}
\newcommand{\SD}{\mbox{SD}}

\newcommand{\surm}{\mathcal{P}}
\newcommand{\sdrm}{\widetilde{\mathcal{P}}}
\newcommand{\sufm}{\mathcal{Q}}
\newcommand{\sdfm}{\widetilde{\mathcal{Q}}}

\newcommand{\FDP}{\mbox{FDP}}

\renewcommand{\l}{\ell}
\newcommand{\TSud}{\mathcal{U}}

\newcommand{\RM}{\mathrm{RM}(m,\pi_0,F)} 
\newcommand{\RMDU}{\mathrm{RM}(m,\pi_0,F\equiv 1)} 
\newcommand{\FM}{\mathrm{FM}(m,m_0,F)} 
\newcommand{\FMDU}{\mathrm{FM}(m,m_0,F\equiv 1)} 
\newcommand{\DU}{\mathrm{DU}(m,m_0)}

\newcommand{\toto}{\l}

\begin{document}
\renewcommand{\baselinestretch}{1.2}
\markright{
}
\markboth{\hfill{\footnotesize\rm GILLES BLANCHARD, THORSTEN DICKHAUS, ETIENNE ROQUAIN AND FANNY VILLERS
}\hfill}
{\hfill {\footnotesize\rm } \hfill}
\renewcommand{\thefootnote}{}
$\ $\par
\fontsize{10.95}{14pt plus.8pt minus .6pt}\selectfont
\vspace{0.8pc}
\centerline{\large\bf ON LEAST FAVORABLE CONFIGURATIONS }
\vspace{2pt}
\centerline{\large\bf FOR STEP-UP-DOWN TESTS}
\vspace{.4cm}
\centerline{Gilles Blanchard, Thorsten Dickhaus, Etienne Roquain and Fanny Villers
}
\vspace{.4cm}
\centerline{\it Potsdam University, Humboldt University and UPMC University}
\vspace{.55cm}
\fontsize{9}{11.5pt plus.8pt minus .6pt}\selectfont

\begin{quotation}
\noindent {\it Abstract:}
This paper investigates an open issue related to false discovery rate (FDR) control of step-up-down (SUD) 
multiple testing procedures. It has been established in earlier literature that for this type of procedure,
under some broad conditions, and in an asymptotical sense,  
the FDR is maximum when the signal strength under the alternative is maximum.
In other words, so-called ``Dirac uniform configurations" are asymptotically {\em least favorable} in this setting. 
It is known that this 
property also holds in a 
non-asymptotical sense (for any finite number of hypotheses), for
the two extreme versions of SUD procedures, namely step-up and step-down (with extra conditions for the step-down case).
It is therefore  very natural to conjecture that this non-asymptotical {\em least favorable configuration} 
property could more generally be true for all ``intermediate'' forms of SUD procedures.
We prove that this is, somewhat surprisingly, not the case.
The argument is based on the exact calculations proposed earlier by \citet{RV2010}, that we extend here by generalizing 
 Steck's recursion to the case of two populations. Secondly, we quantify the magnitude of this phenomenon by providing 
a nonasymptotic upper-bound and explicit vanishing rates as a function of the
total number of hypotheses.
\par

\vspace{9pt}
\par
\end{quotation}\par

\fontsize{10.95}{14pt plus.8pt minus .6pt}\selectfont
\section{Introduction}\label{sec:intro}

In mathematical statistics, so-called least favorable parameter configurations (LFCs) play a pivotal role. 
For a given statistical decision problem over a parameter space $\Theta$
and a given decision rule $\delta$, we define
an LFC $\theta^*(\delta)$ as any element of $\Theta$ that maximizes the risk (expected loss) of
$\delta$ under this parameter, i.e.,
$$
\forall \theta \in \Theta,\: R(\theta, \delta) \leq R(\theta^*(\delta), \delta),
$$
where $R(\theta, \delta)$ denotes the risk 
of rule $\delta$ under $\theta$. When available, the knowledge of an LFCs allows one to obtain a bound on the risk 
over a possibly very large parameter space, including non- or semi-parametric cases where $\Theta$ has infinite dimensionality. Theoretical investigations of minimax properties can rely on the computation of an LFC. 
Such knowledge is also relevant for practice, because a user of the procedure can be provided with a performance guarantee if an LFC is known. In this case, even if the risk under the LFC cannot be computed in closed form, it can be approximated by a Monte-Carlo method simulating the distribution corresponding to the LFC.
Finally, if the parameter space $\Theta$ is partitioned into disjoint, restricted submodels,
it can be of interest to gain knowledge of the LFC of a decision rule $\delta$ separately over each
submodel, thus providing finer-grained information.


LFC considerations naturally occur in hypothesis testing problems. 
A classical example is that of one-sided tests over a one-dimensional parameter space 
admitting an isotonic likelihood ratio: it is well-known that the LFC
for the type I error probability is located at the boundary of the null hypothesis. This fact is used to derive critical values for uniformly more powerful tests in that setting.

The LFC problem is particularly delicate for multiple hypotheses testing and the latter has been investigated by many authors in previous literature \citep{FR2001,BY2001,LR2005,FDR2007,RW2007,GR2008,SH2008,FDR2009,FG2009,Gont2010}. In that setting, a family of $m \geq 2$ null hypotheses $H_1,\ldots,H_m$ is to be tested simultaneously under the scope of a common statistical model with parameter space $\Theta$, and some type I error criterion is used that accounts for multiplicity. 
For theoretical as well as practical applications, it is
relevant  to determine LFCs over the restricted parameter spaces $\Theta_{m,m_0}$ where exactly $m_0$ out of $m$ of the null hypotheses are 
true.
In this setting, LFC results can be derived straightforwardly only in special situations. In the present work, we restrict our attention to multiple testing procedures that depend on the observed data only through a collection of marginal $p$-values, each associated to an individual null hypothesis. This is a commonly used setting for multiple testing problems in high dimension.
Moreover, we consider procedures that reject exactly those null hypotheses having their 
$p$-value less than a certain common threshold $t^*$, which can possibly be data-dependent.
That is, $t^*$ may depend in a complex way of the entire family of $p$-values. We call
such procedures threshold-based for short.

In this setting, 
 LFCs crucially depend on the type I error
criterion considered. One frequently encountered family of such criteria is given through loss functions that
only depend on the number of type I errors, denoted
\begin{equation}
\label{eq:defV}
V_m=V_m(\theta,\delta) 
:= |\{1\leq i \leq m: H_i \text{~~is true for $\theta$ and gets rejected by $\delta$}\}|.
\end{equation}
In other words, the risk takes the form $R(\theta,\delta)=\E_\theta[\phi(V_m)]$.
Natural assumptions are that $t^*$ is a nonincreasing function in each $p$-value and that
$\phi$ is a nondecreasing function.
Then, by additionally assuming that the $p$-values are jointly independent,  it is known that the LFC over $\Theta_{m,m_0}$ 
is a {\em Dirac-uniform} (DU) distribution \citep[introduced by][]{FR2001}, 
i.\,e.\,, such that $p$-values corresponding
to true nulls are independent uniform variables, while $p$-values 
under alternatives follow a Dirac distribution with point mass $1$ in zero. 
This result is formally restated in Appendix~\ref{sec:FWER_LFC}.
For example, this LFC property holds true under the above assumptions for the $k$-family-wise error rate ($k$-FWER). For a given $\theta \in \Theta$, the $k$-FWER under  $\theta$ is defined by $\text{FWER}_{k,\theta} := \mathbb{P}_\theta(V_m \geq k)$. Strong control of the (1-)family-wise error rate, i.~e., ensuring that $\sup_\theta \text{FWER}_{1,\theta} \leq \alpha$ for a pre-defined level $\alpha \in (0, 1)$, is the usual type I error concept in traditional multiple hypotheses testing theory. 

However, over the last two decades, progress in application fields such as genomics, proteomics, neuroimaging, and astronomy has lead to massive multiple testing problems with very large systems of hypotheses \citep{DL2008,PNBL2005,Astro2001}. In this type of applications, ($k$-)FWER control is typically too strict a requirement, and a less stringent notion of type I error control is needed in order to ensure reasonable power of corresponding multiple tests. In particular, the false discovery rate (FDR) 
introduced by \citet{BH1995} has become a standard criterion for type I error control in large-scale multiple testing problems. The FDR is defined as the expected proportion of type I errors among all rejections.
Unfortunately, it does not fall into the class of type I error measures considered in the previous paragraph, so that the above result does not apply. Furthermore, because the average ratio of two dependent random variables is not necessarily increasing in the value of the numerator,
the LFC problem for the FDR criterion turns out to be a challenging issue -- even for simple classes of multiple tests and under independence assumptions.

In this work, we contribute to the theory of LFCs under the FDR criterion for so-called step-up-down multiple tests (SUD procedures, for short). These procedures constitute an important and wide subclass of threshold-based multiple testing procedures, wherein the threshold $t^*$ is obtained by comparing the reordered $p$-values to a fixed set of critical values \citep[see][]{TLD1998,Sar2002}. 
Furthermore, recent research has reinforced the interest of this type of procedures. For instance, 
\citet{FDR2009} have shown that step-up-down tests can be used is association with the so-called
asymptotically optimal rejection curve (AORC) to provide an asymptotically (as $m\rightarrow \infty$) valid FDR control which is additionally optimal in some specific sense.

Namely, the contributions of the paper are as follows:
\begin{itemize}
\item a survey of known LFC results for SUD procedures in specific model classes is provided in Section~\ref{sec:surveyLFC};
\item new results and surprising counterexamples for LFCs of SUD procedures are derived in Section~\ref{sec:newres}.
\item in Section~\ref{sec:exactformulas}, Steck's recursion is extended to the case of two populations and we provide a summary of the exact formulas for computing the FDR proposed by \citet{RV2010,RV2010sup}; these formulas are used to derive the counterexamples previously mentioned.
\end{itemize}

\section{Mathematical setting}

\subsection{Models}

Given a statistical model, we consider a finite set of $m\geq 2$ null hypotheses $H_1,\ldots,H_m$,
and a corresponding, fixed collection of tests with associated $p$-value family
$\mbf{p}:=(p_i, i\in\{1,...,m\})$. 
 For simplicity, we skip somewhat the formal definition of $p$-values and of the underlying statistical
model and consider directly a statistical model for the $p$-values, that is, a model for the 
joint distribution of $\mbf{p}$. 
In what follows, we denote by $\mathcal{F}$ the set containing c.d.f.'s from $[0,1]$ into $[0,1]$ 
that are continuous. 

\begin{itemize}
\item[\textbullet]
The $p$-value family $\mbf{p}$ follows the \textit{(two group) fixed mixture model} with parameters $m\geq 2$, $1\leq m_0\leq m $ and $F\in \mathcal{F}$, for which the corresponding distribution is denoted by  $\mbox{FM}{(m,m_0,F)}$, if $\mbf{p}=(p_i, i\in\{1,...,m\})$ is a family of mutually independent variables and for all $i$, 
$$p_i \sim \left\{\begin{array}{ll} U(0,1) & \mbox{ if } 1\leq i \leq m_0, \\  F &\mbox{ if } m_0+1\leq i \leq m,\end{array}\right.$$
where $U(0,1)$ denotes the uniform distribution on $(0,1)$.
\item[\textbullet]
The $p$-value family $\mbf{p}$ follows the \textit{(two group) random mixture model} with parameters $m\geq 2$, $\pi_0\in[0,1]$ and $F\in \mathcal{F}$, for which the corresponding distribution is denoted by  $\mbox{RM}{(m,\pi_0,F)}$, if there is an (unobserved) binomial random variable $m_0\sim\mathcal{B}(m,\pi_0)$ such that  $\mbf{p}$
follows the $\mbox{FM}{(m,m_0,F)}$ model conditionally on $m_0$. In that case, the $p$-values are i.i.d. with (unconditional) c.d.f. $G(t)=\pi_0 t + (1-\pi_0) F(t)$.
\end{itemize}

In the above definition, note that the true nulls are automatically assigned to the $m_0$ (random or not) first coordinates. This can be assumed without loss of generality, since we only consider procedures which only depend on $p$-values through their reordering in increasing order. 

A common additional assumption on $F$ is that $F(x)\geq x$ or that $F$ is concave. 
For instance, these assumptions are both satisfied in the two following standard examples:
\begin{itemize}
\item[-] 
Gaussian location model: $F(t)=\ol{\Phi}(\ol{\Phi}^{-1}(t)-\mu)$, for a given alternative mean $\mu>0$, where $\ol{\Phi}(z)=\P(Z\geq z)$ for $Z\sim\mathcal{N}(0,1)$. This corresponds to the alternative distribution of $p$-values
when testing for $\mu\leq 0$ under a Gaussian location shift model with unit variance.
\item[-] Dirac $\delta_0$ distribution: $F$ is identically equal to $1$, as introduced by \citet{FR2001}. The corresponding distribution in the FM model is called Dirac-uniform (DU) configuration (or distribution) and denoted by 
$\FMDU$ or simply $\DU$. We define similarly  $\RMDU$. Note that the Dirac-uniform configuration can be seen as an instance of the 
Gaussian c.d.f. for an alternative mean $\mu=\infty$.
\end{itemize}
In the existing literature, the Dirac-uniform distribution has often be considered as the first candidate for being an LFC of several global type I error rates (with or without a theoretical support) \citep[see, e.g.,][]{FDR2007,RW2007,SH2008}.

\subsection{Procedures} \label{sec_proced}

 In this paper, we consider the particular class of multiple testing procedures called step-up-down procedures, 
introduced by \citet{TLD1998}, see also \citet{Sar2002}.
First define a \textit{threshold} or \textit{critical value} collection as any nondecreasing sequence $\mbf{t}= (t_k)_{1\leq k \leq m} \in [0,1]^m$ (with $t_0=0$ by convention).
\begin{definition}
 Let us order the $p$-values $p_{(1)}\leq p_{(2)} \leq ...\leq p_{(m)}$ (with the convention $p_{(0)}=0$). 
For any threshold collection $\mbf{t}$, the \textit{step-up-down (SUD) procedure} with threshold 
collection $\mbf{t}$ and of order
$\lambda\in\{1,...,m\}$, denoted here by ${\SUD}(\mbf{t})$, rejects the $i$-th hypothesis if $p_i\leq t_{\hat{k}}$, with 

\begin{equation}\label{equ-defSUD}
\widehat{k}= \left\{\begin{array}{ll}\max\{k\in\{\lambda,\dots,m\}\telque \forall  k'\in\{\lambda,\dots,k\}, \:p_{(k')}\leq {t}_{k'}\} & \mbox{ if } p_{(\lambda)} \leq t_{\lambda};\\  \max\{k\in\{0,\dots,\lambda\}\telque \:p_{(k)}\leq {t}_{k}\}&\mbox{ if }p_{(\lambda)} > t_{\lambda}. \end{array} \right.
 \end{equation}
\end{definition}

In the sequel, for convenience, we identify procedures with their rejection sets, e.g., ${\SUD}(\mbf{t})=\{1\leq i \leq m\telque p_i\leq t_{\hat{k}}\}$.
An important remark is that the cases $\lambda=1$ and $\lambda=m$ correspond to the traditional step-down and step-up procedures, respectively. An illustration is provided in Figure~\ref{fig:sud}.

\begin{figure}[h!]
\includegraphics[scale=0.17]{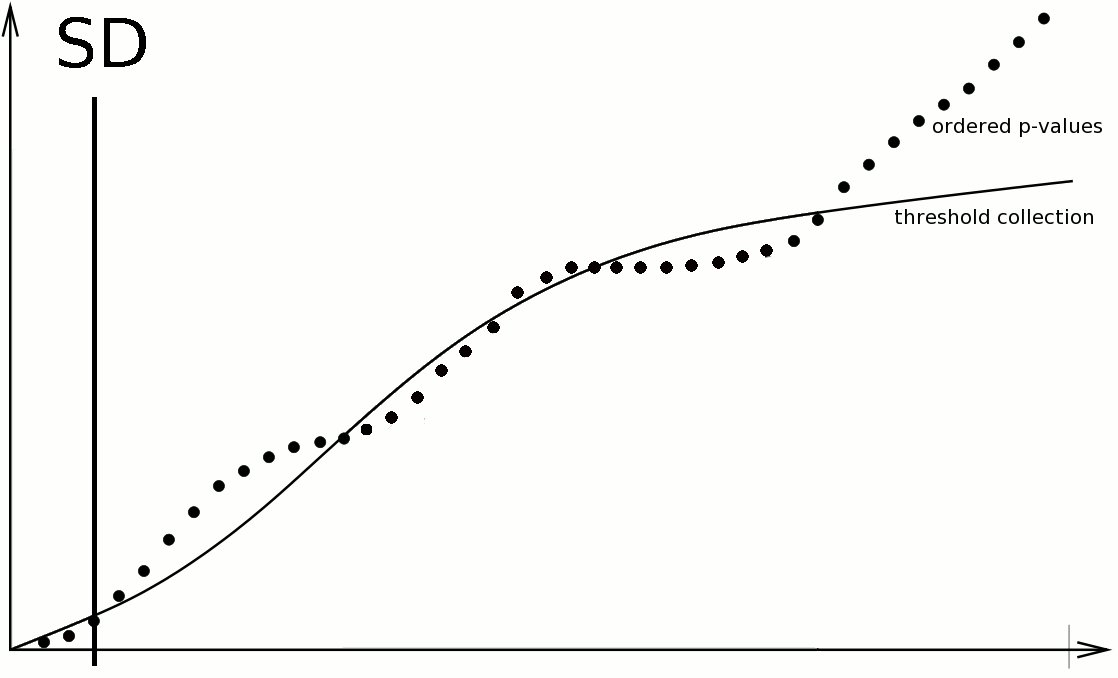}\includegraphics[scale=0.17]{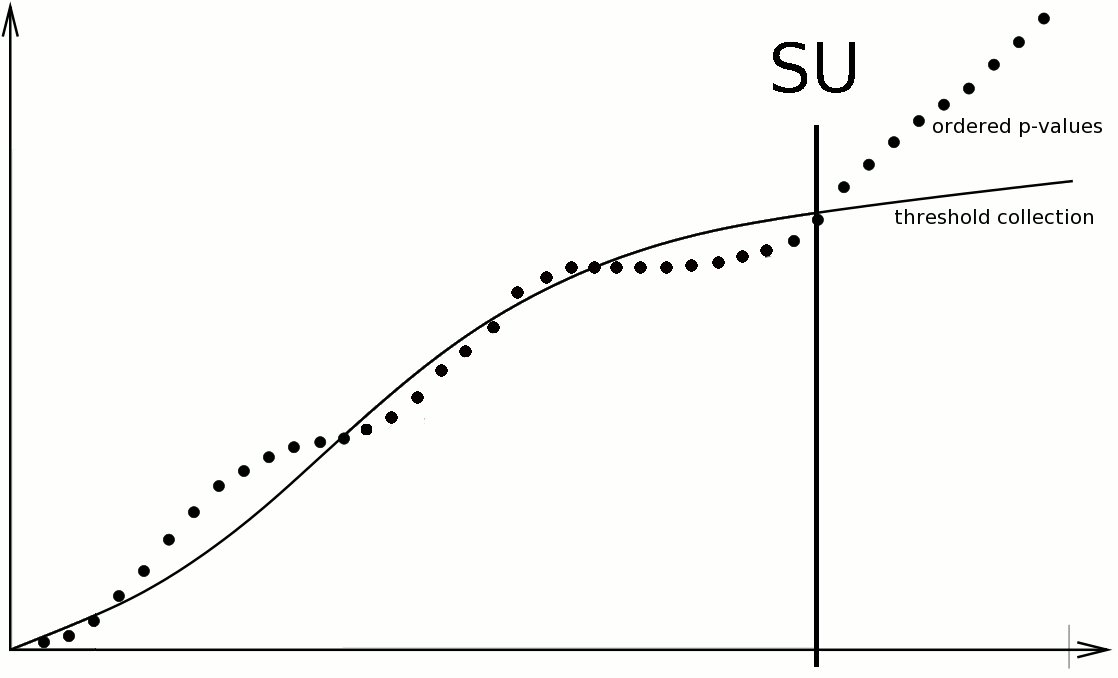}\\
\includegraphics[scale=0.17]{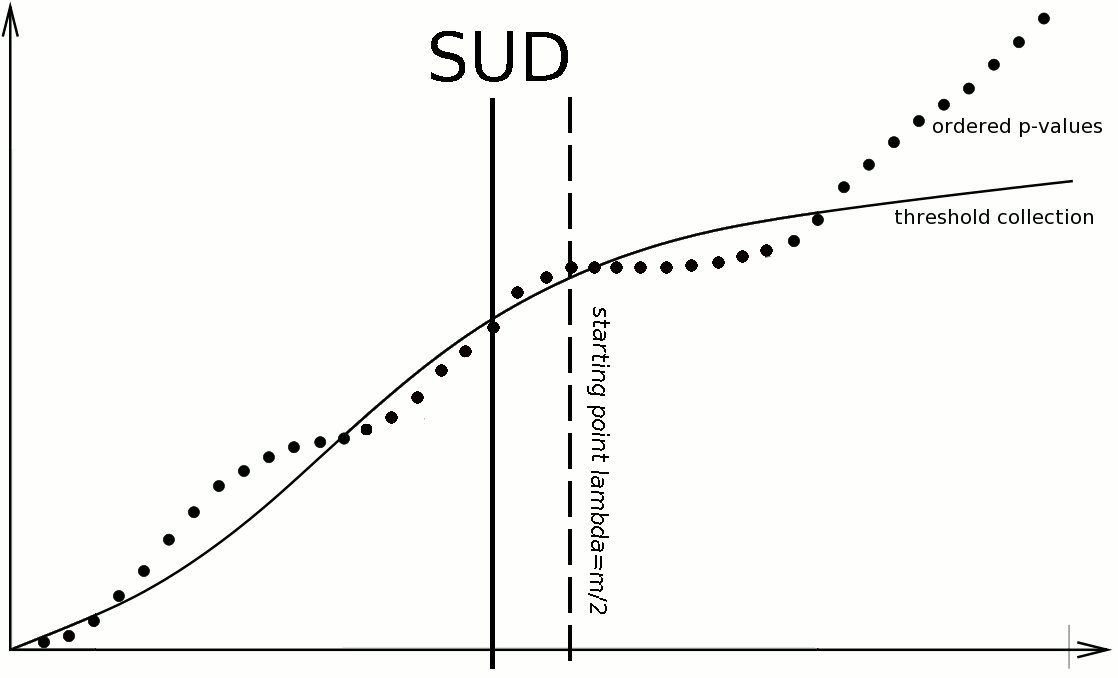}\includegraphics[scale=0.17]{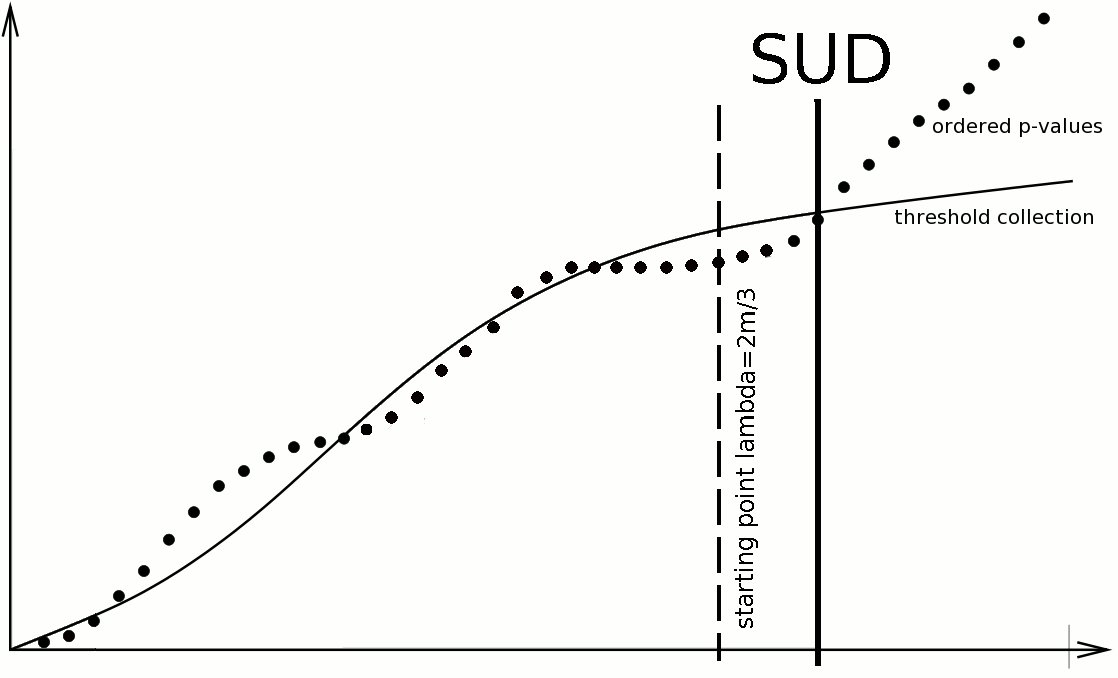}
\caption{Value of $\widehat{k}$ (vertical solid line), defined by \eqref{equ-defSUD},  for several procedures of the SUD type. The bottom-right SUD procedure (using $\lambda=2m/3$) coincides with the SU procedure for this realization of the $p$-value family.\label{fig:sud}}
\end{figure}

A classical choice for the threshold collection $\mbf{t}$ consists of Simes' \citeyearpar{Sim1986} critical values $t_k=\alpha k/m$ for a pre-specified level $\alpha\in(0,1)$. The corresponding step-up-down procedure is called  the \textit{linear step-up-down} procedure and is denoted by $\LSUD$.
In particular,  for $\lambda=1$ and $\lambda=m$, the procedure $\LSUD$ is simply denoted by $\LSD$ and $\LSU$, respectively. $\LSU$ corresponds to the famous linear step-up procedure of \citet{BH1995}. 

It is common to  consider threshold collections of the form
$t_k=\rho(k/m)$ for a function $\rho:[0,1]\rightarrow [0,1]$. This function
is generally assumed to satisfy the following assumptions:
\begin{align}\label{rhocond1}
&\mbox{$\rho:[0,1]\rightarrow [0,1]$ is continuous and non-decreasing;}\\&\mbox{$x\in(0,1]\rightarrow \rho(x)/x$ is non-decreasing.}\label{rhocond2}
\end{align}
The function $\rho$ is called the critical value function (while its inverse is generally called the
rejection curve, see e.g. \citealp{FDR2009}).
Observe that assumption \eqref{rhocond1} can always be assumed to hold when $m$ is fixed,
however it is of interest in the case of an asymptotical analysis as $m\rightarrow \infty$ (in which
case $\rho$ is assumed to be independent of $m$).
Assumption \eqref{rhocond2} on the other hand restricts the possible threshold collection also for  
any fixed $m$.  It will be often used in this paper. For a fixed finite $m$, assumptions 
\eqref{rhocond1} and \eqref{rhocond2} taken together are equivalent to ``$k\mapsto t_k/k$ is non-decreasing".


\subsection{False discovery rate and LFCs}

As introduced by \citet{BH1995}, the false discovery rate of a multiple testing procedure is defined as the averaged ratio of the number of erroneous rejections to the total number of rejections. In our setting, for a distribution $P$ being either $\FM$ or $\RM$, the FDR of a step-up-down procedure can be written as
\begin{align}
\label{equ_FDR_FM}
\FDR(\SUD(\mbf{t}),P)&=\E_{\mathbf{p}\sim P} \left[ \FDP(\SUD(\mbf{t}),m_0,\mbf{p}) \right],
\end{align}
for which the FDP is the false discovery proportion defined by
\begin{align}
\label{equ_FDP_FM}
\FDP(\SUD(\mbf{t}),m_0,\mbf{p})&= \frac{ |\{1\leq i \leq m_0 : p_i\leq  t_{\hat{k}} \}|}{ |\{1\leq i \leq m : p_i\leq  t_{\hat{k}} \}|\vee 1},
\end{align}
where $|\cdot|$ denotes the cardinality function and in which $m_0$ is either fixed or random whether $P$ is $\FM$ or $\RM$, respectively. For short, the quantity $\FDR(\SUD(\mbf{t}),\FM)$ is often denoted $\FDR(\SUD(\mbf{t}),m_0,F)$, or $\FDR(\SUD(\mbf{t}),F)$ when the context makes the interpretation unambiguous. Similarly,  $\FDR(\SUD(\mbf{t}),\RM)$ can be shortened as $\FDR(\SUD(\mbf{t}),\pi_0,F)$, or $\FDR(\SUD(\mbf{t}),F)$.

\begin{definition}\label{def:LFC}
Any $F'  \in \mathcal{F}$ is called a least favorable configuration (LFC) for the FDR of $\SUD(\mbf{t})$ in a fixed mixture model with $m_0$ true hypotheses out of $m$ and  over the class $\mathcal{F}$
if 
$$
\forall F\in \mathcal{F},\:\: \FDR(\SUD(\mbf{t}),m_0,F) \leq \FDR(\SUD(\mbf{t}),m_0,F').
$$
A similar definition holds for a random mixture model with $m$ hypotheses and 
proportion $\pi_0$ of true hypotheses. 
\end{definition}
The above definition can possibly be restricted to a subclass 
$\mathcal{F}' \subset \mathcal{F}$ (typically, the class of concave c.d.f.s).
This will be clearly specified in the context.

Obviously, if $F'$ is an LFC for the fixed mixture model for all values of $m_0$,
then it is also an LFC in the $\RM$ model for any value of $\pi_0$
(by integrating over  $m_0\sim\mathcal{B}(m,\pi_0)$).

\section{Survey of known LFCs under the FDR criterion}\label{sec:surveyLFC}

Recent results about LFCs for the FDR criterion
related to step-up-down type procedures are summarized in Figure~\ref{fig_FDR} and explained below. 
These results hold either under the fixed mixture model or the random mixture model, hence involve a maximization over distributions where the $p$-values are independent. (While the present paper is focused on this setting, 
let us mention here briefly, that LFCs for the FDR criterion under {\em arbitrary} dependencies 
have also been studied, see, e.g. \citealp{LR2005,GR2008}.)

First, let us consider the problem of the monotonicity of $\FDR(\mbox{SUD}_{\lambda}(\mbf{t}))$ in $\lambda$ (vertical arrows). 
Recently, it was proved that, whenever $F$ is concave, the FDR grows as the rejection set grows \citep[Theorem~4.1]{ZZD2011}. Interestingly, the rejection set $R$ can have a very general form: the only condition is that $|R|$ is a measurable function of the order statistics of the family of $p$-values under consideration. From \eqref{equ-defSUD} and since for any $\lambda\in\{1,\dots,m-1\}$, 
the rejection set of $\mbox{SUD}_{\lambda}(\mbf{t})$ is included in the one of $\mbox{SUD}_{\lambda+1}(\mbf{t})$, we obtain that for a concave $F$,  $$\FDR(\mbox{SUD}_{\lambda}(\mbf{t}))\leq \FDR( \mbox{SUD}_{\lambda+1}(\mbf{t})),$$ both for $\FM$ and $\RM$ models. This implies in particular that $\FDR(\mbox{SD}(\mbf{t}))\leq \FDR( \mbox{SU}(\mbf{t}))$ for a concave $F$.
Other studies establish similar inequalities, but with a condition on the threshold collection 
$\mbf{t}$, not on $F$. 
Precisely,  Theorem~4.3 of \citet{FDR2009} and Theorem~3.10 of \citet{Gont2010} establish that, when $k\mapsto t_k/k$ is nondecreasing, for any $\lambda\in\{1,\dots,m-1\}$, 
$$\FDR(\mbox{SUD}_{\lambda}(\mbf{t}))\leq \FDR( \mbox{SU}(\mbf{t})),$$
both for $\FM$ and $\RM$ models.
In particular, the fact that $\FDR( \mbox{SU}(\mbf{t}))$ dominates the FDR of  $ \mbox{SD}(\mbf{t})$ 
 is quite well established in multiple testing literature. Nevertheless, let us stress that this is no longer the case for ``atypical'' 
configurations of $F$ and $\mbf{t}$, as we state in Appendix~\ref{result-extrem}.

Secondly, let us consider the monotonicity of $\FDR(\mbox{SUD}_{\lambda}(\mbf{t}),F)$ in $F$. 
In the step-up case (i.e., $\lambda=m$), the situation is somewhat simple: Theorem~5.3 of \citet{BY2001} states that $F\leq F'$ implies  $\FDR(\mbox{SUD}_{\lambda}(\mbf{t}),F) \leq \FDR(\mbox{SUD}_{\lambda}(\mbf{t}),F')$ whenever $k\mapsto t_k/k$ is nondecreasing. Moreover, the inequality is reversed whenever  $k\mapsto t_k/k$ is nonincreasing.
In the step-down case (i.e., $\lambda=1$) and for a $\RM$ model, Theorem 4.1 of \citet{RV2010} states that the Dirac-uniform configuration ($F\equiv 1$) is an LFC under some complex condition on the threshold 
collection $\mbf{t}$, that is fulfilled by the linear threshold collection $t_k=\alpha k/m$, $\alpha\in(0,1)$ and over the class of concave c.d.f.'s.  However, for $\lambda\notin \{1,m\}$ (i.e., an ``intermediate" SUD procedure), finding LFC's is more delicate and the only known result is asymptotic, as $m$ tends to infinity.
Precisely, combining  Theorem~4.3 of \citet{FDR2009} and Lemma~3.7 of \citet{Gont2010}, we easily derive the following result: 
\begin{theorem}\label{th-Gont}[\citet{Gont2010}]
Consider a step-up-down procedure using a threshold collection of the form 
$t_k=\rho(k/m)$, where $\rho$ satisfies \eqref{rhocond1} and \eqref{rhocond2}.
Assume that the step-up-down procedure is performed at an order $\lambda=\lambda_m$ such that $\lambda_m/m\rightarrow \kappa\in[0,1]$.
Assume that $m_0/m\rightarrow\zeta\in[0,1]$ and that, 
under the $\DU$ distribution, the number of rejections of $\SUD(\mbf{t})$ satisfies that $|\SUD(\mbf{t})|/m$ converges in probability as $m$ grows to infinity.
Then, in the fixed mixture model $\FM$, we have for any $F\in\mathcal{F}$,

\begin{equation}\label{equ_asymp}
\limsup_m\big\{ \FDR(\SUDm(\mbf{t}),F) - \FDR(\SUDm(\mbf{t}),F\equiv 1)\big\}\leq 0,\end{equation}
either for all $\zeta\in[0,1]$ if $\kappa>0$ or for all $\zeta\in[0,1)$ if $\kappa=0$.
\end{theorem}
However, for a finite $m$, and $\lambda\notin \{1,m\}$ no result is known about LFC's to our knowledge. This is the point of the paper and is symbolized by the question mark in the middle of Figure~\ref{fig_FDR}. 

Finally, let us consider the linear SUD procedure, that is, the SUD procedure using the threshold collection
$t_k=\alpha k/m$, $\alpha\in(0,1)$ (corresponding to $\rho(x)=\alpha x$).
Since both LSU and LSD procedures satisfy that DU is an LFC and since an SUD procedure can be expressed as a combination of an SU and an SD, we might make the following conjecture, which is the starting point of this paper. 
\begin{conj}\label{main-conj}
For any $m\geq 2$, the Dirac-uniform configuration ($F\equiv1$) is a least favorable configuration for the FDR of the linear step-up-down procedure, in the $\RM$ and $\FM$ models.
\end{conj}

Obviously, a similar conjecture might be formulated for a (non-linear) step-up-down procedure using $\rho$ satisfying \eqref{rhocond1} and \eqref{rhocond2}.

 \begin{figure}[htbp]
\begin{center}
\includegraphics[scale=0.7]{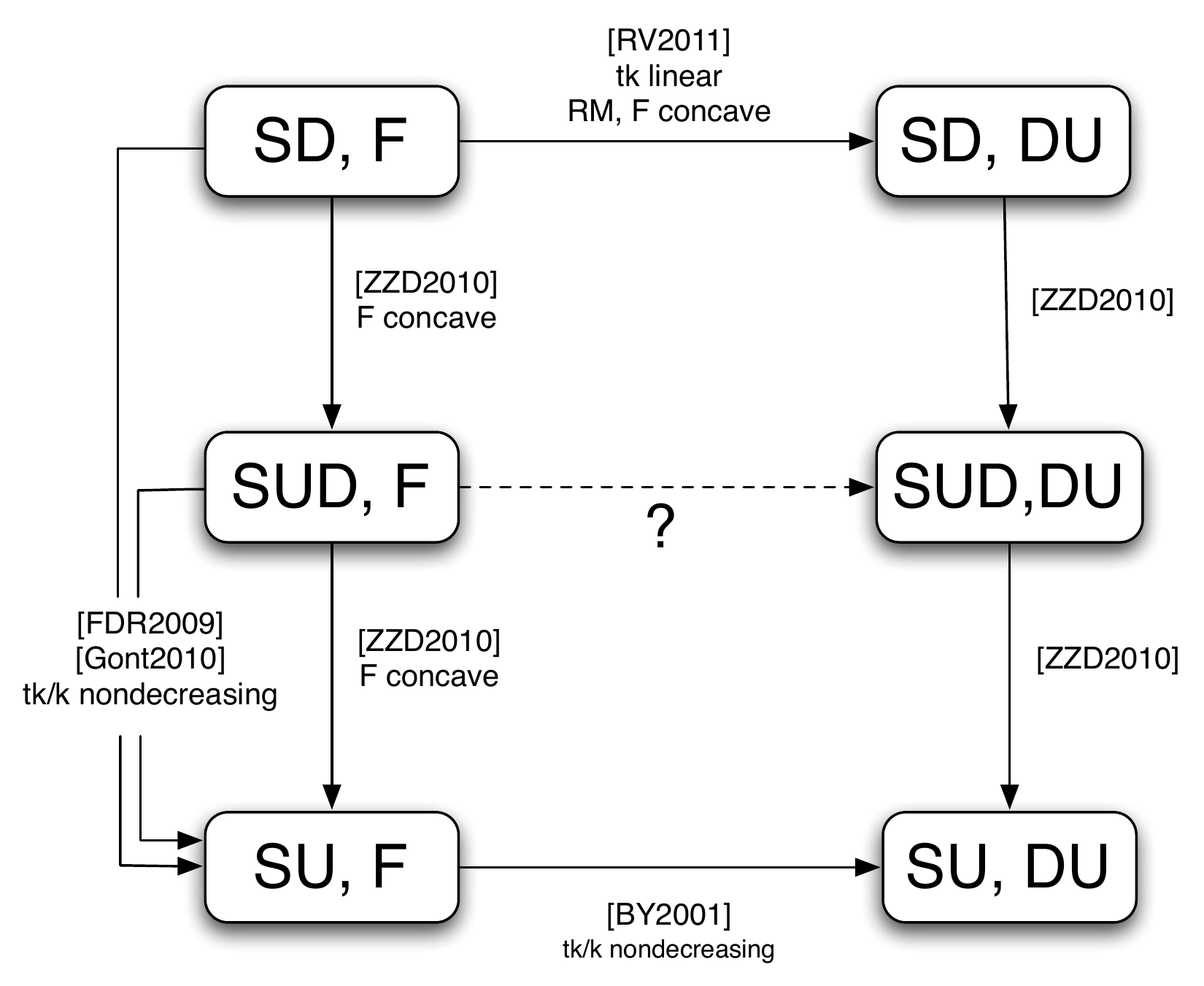}
\caption{An arrow ``A  $\rightarrow$ B" means ``$\FDR(A) \leq \FDR(B)$". 
These results hold for the fixed mixture (FM) model except when ``RM" is written. The brackets are a shortened reference to the corresponding 
literature, see main text for more details.}
\label{fig_FDR}
\end{center}
\end{figure}

\section{Investigating Conjecture~\ref{main-conj}}\label{sec:newres}


\subsection{Disproving the conjecture}

The exact calculations described in Section~\ref{sec:exactformulas} allow to compute the value of $\FDR(\mbox{LSUD}_{\lambda}(\mbf{t}))$ exactly. 
This shows the following (numerical) result.

\begin{proposition}\label{prop:disproof}
For $m =10$, consider the linear step-up-down procedure $\LSUD$ at  level $\alpha=0.5$ and for $\lambda \in \left\{ 4,5,6,7\right\}$. Then, we have 
\begin{align}
 \FDR(\LSUD,F) > \FDR(\LSUD,F\equiv1) , \label{equ-disproof}
  \end{align}
in either of the two following cases:
\begin{itemize}
\item[\textbullet] in the $\FM$ model, with $m_0=7$ and $F(x)=x$;
\item[\textbullet] in the $\RM$ model, with $\pi_0=7/10$ and $F(x)=x$.
\end{itemize}
\end{proposition}

This disproves Conjecture~\ref{main-conj} and shows that finding the LFC for SUD is more difficult than for SU and SD separately.
More generally, Figure~\ref{fig:LSUD-LFC} reports some obtained values in the Gaussian case $F=\ol{\Phi}(\ol{\Phi}^{-1}(\cdot)-\mu)$ for different values of the alternative mean $\mu$. 
We observe that the result for $\FM$ (left panel) or $\RM$ (right panel) are qualitatively the same: there is a range of values for $\lambda\in\{2,\dots,m-1\}$ for which the FDR is larger for a smaller $\mu$. However, this phenomenon seems to decay when $m$ becomes larger, see Figure~\ref{fig:LSUD-LFC} for $m=100$. Also, when $\alpha$ decreases, the phenomenon still occurs but its amplitude decays. 
\begin{figure}[h!]
\begin{tabular}{ccc}
& Fixed mixture $m_0/m=0.7$ & Random mixture $\pi_0=0.7$\\ 
\rotatebox{90}{\hspace{2.5cm} $m=10$} & \includegraphics[scale=0.35]{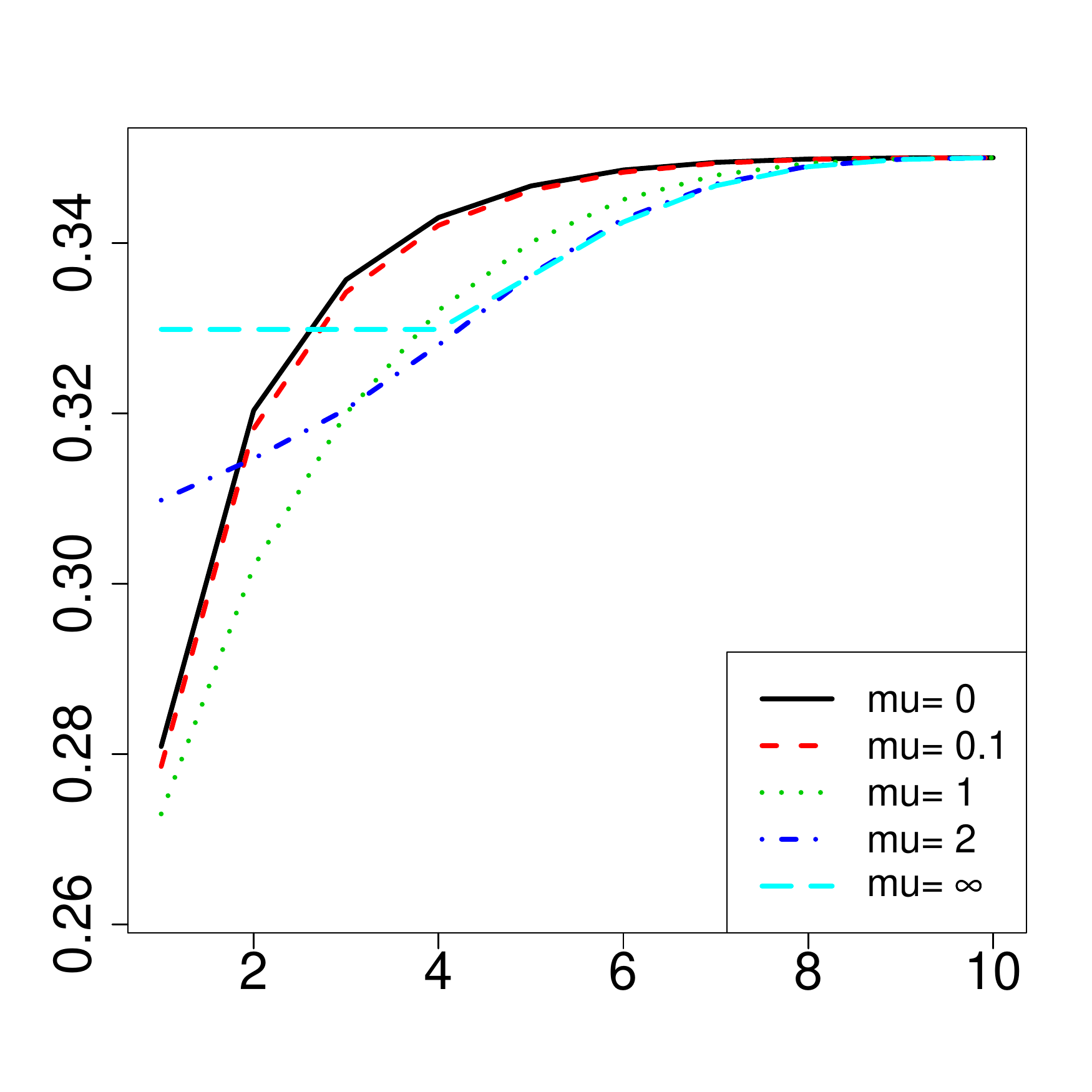}&\includegraphics[scale=0.35]{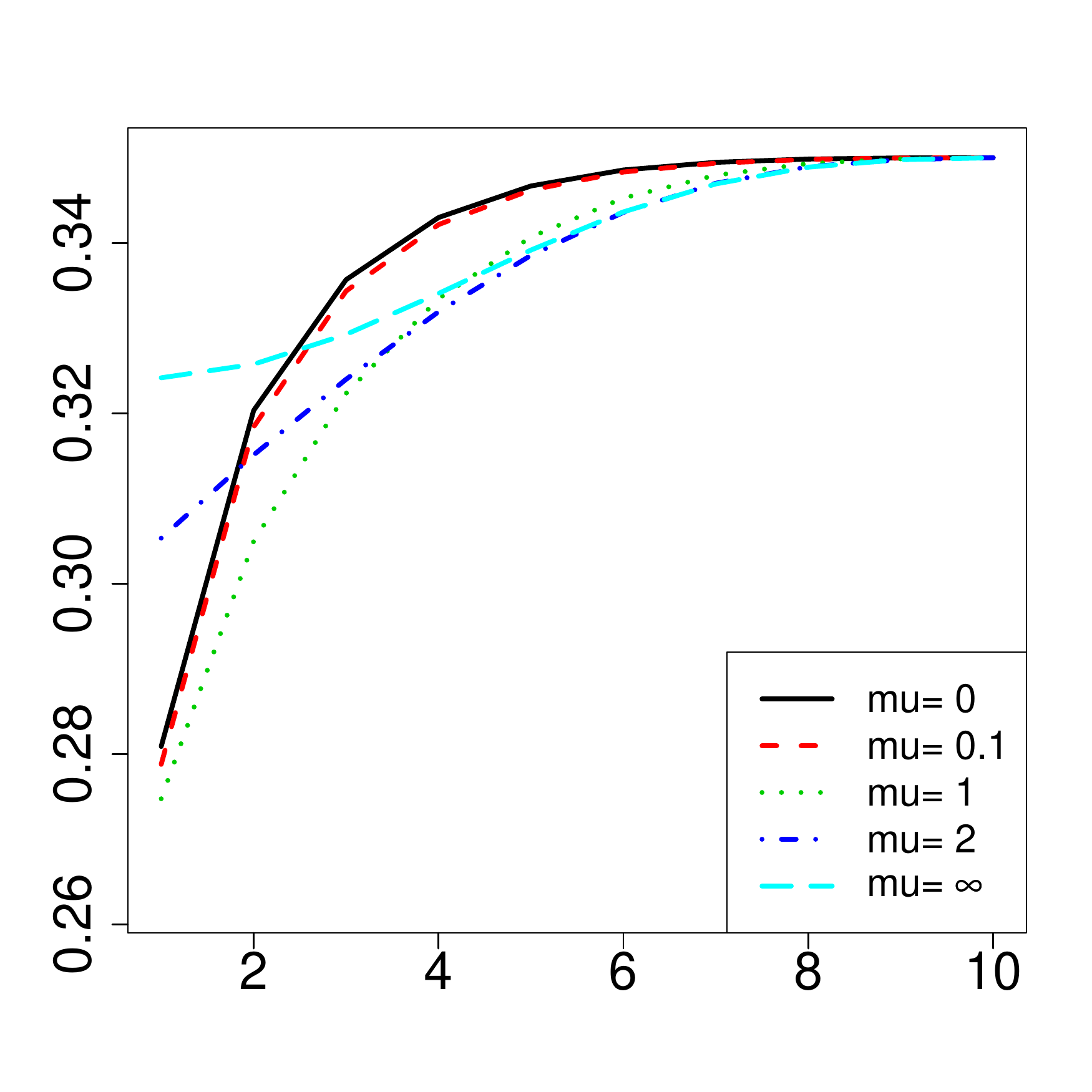}\\
\rotatebox{90}{\hspace{2.5cm} $m=100$} & \includegraphics[scale=0.35]{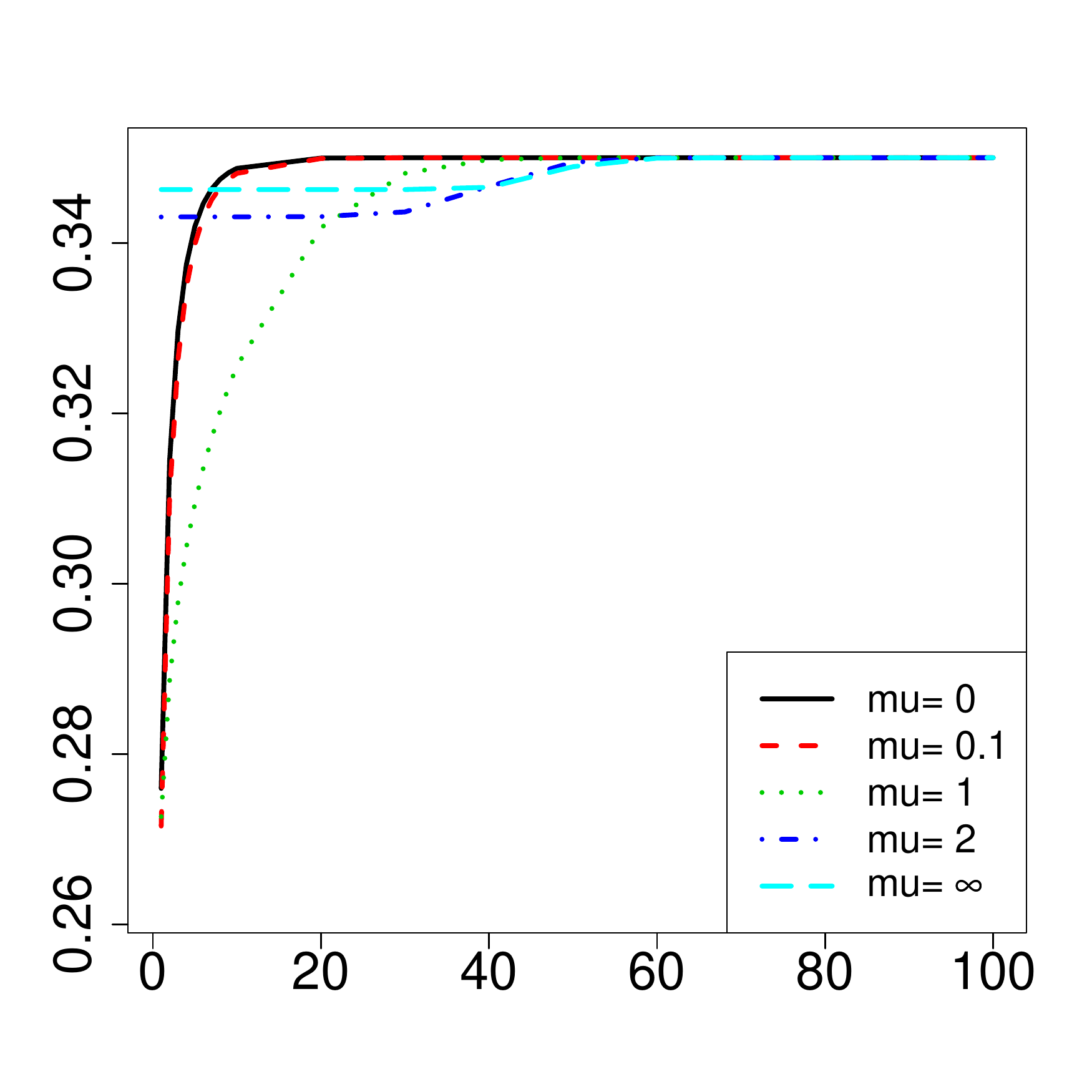}&\includegraphics[scale=0.35]{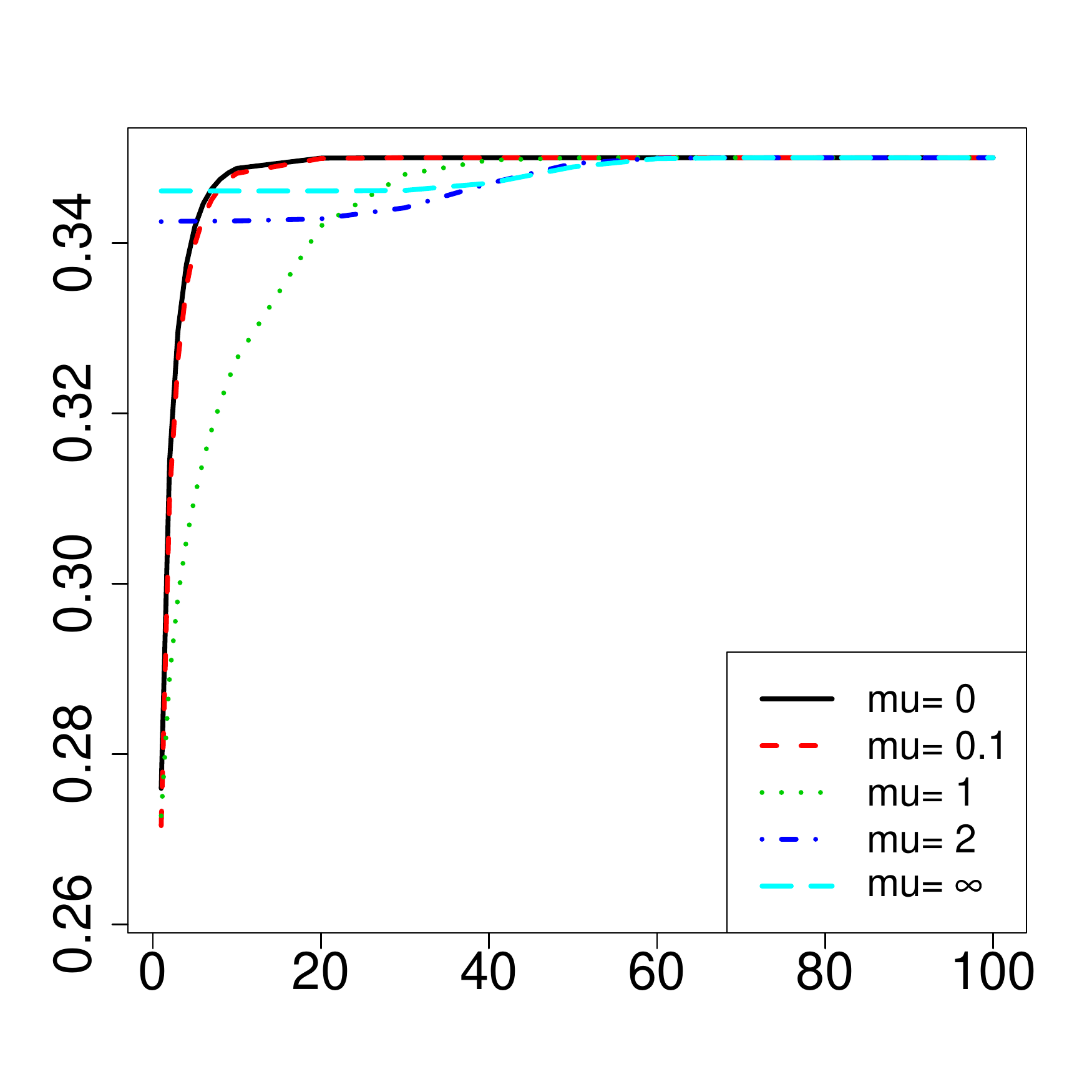}
\end{tabular}
\caption{The LFC of LSUD is not always DU. $\FDR(\LSUD)$ as a function of the order $\lambda\in\{1,\dots,m\}$
. Left: fixed mixture; Right: random mixture. One sided Gaussian location model with parameter $\mu$.  $\alpha=0.5$. \label{fig:LSUD-LFC}}
\end{figure}
To alleviate the concern that this somewhat unexpected phenomenon could be due to numerical inaccuracies
in the computation of the exact formulas (which involve several nested recursions), the
reported results were double-checked via extensive and independent Monte-Carlo simulations,
which confirmed the validity of the reported curves.

\subsection{Nonasymptotic bound}\label{sec:finite}

In the present section, we investigate the amplitude of the
phenomenon observed in the previous section as a function of the number of hypotheses. In other words,
we derive a more explicit and non-asymptotic version of the limit
appearing in \eqref{equ_asymp}. For this, we consider a perturbation
analysis of the SUD procedure as defined in \eqref{equ-defSUD},
under the Dirac-uniform model, and
when the empirical c.d.f. of the $p$-values is $\delta$-close
to the population c.d.f. (which happens with large probability).
In order to state the result in a compact form, we first
introduce the following notation for an SUD threshold in a continuous setting.

\begin{definition}
Let $\rho:[0,1]\rightarrow [0,1]$ satisfying \eqref{rhocond1}. 
For any non-decreasing function
$G:[0,1] \rightarrow [0,1]$, and $\toto \in [0,1]$, we define
\begin{equation} \label{equ-Udef}
\TSud(\toto,G) = 
\begin{cases}
\min \set{u \in [\toto, 1] \telque G(\rho(u)) \leq u } & \text{ if }
G(\rho(\toto)) \geq \toto;\\
\max \set{u \in [0,\toto] \telque G(\rho(u)) \geq u } & \text{ if }
G(\rho(\toto)) < \toto.
\end{cases}
\end{equation}
 \end{definition}
Observe that in the above definition, the infimum and supremum are
well-defined since the considered sets are non-empty; 
using that $G$ is non-decreasing, it can be seen that
$\TSud(\toto,G)$ is a fixed point of the function $G\circ\rho$
(so that the infimum is indeed a minimum and the supremum, a maximum).
Unfortunately, the number of rejections $ \wh{k}$ of the SUD procedure as defined in \eqref{equ-defSUD} 
does not always satisfy $ \wh{k}/m=\TSud(\lambda/m,\hat{\mathbb{G}}_m)$ (because of the step-down part, see Figure~\ref{fig:contrex-SUD} in Section~\ref{proof-lemma-fund}).
Nevertheless, the following lemma is proved in Section~\ref{proof-lemma-fund}.

\begin{lemma}\label{lemma-fund}
With the above notation,  if the threshold collection $\mbf{t}$ is defined as $t_k = \rho(k/m)$, we have
\begin{equation}\label{equ-fund}
\TSud(\lambda/m,\hat{\mathbb{G}}_m) \leq \wh{k}/m \leq \TSud(\lambda/m,(\hat{\mathbb{G}}_m+m^{-1})\wedge 1),
\end{equation}
where $\wh{k}$ is defined by \eqref{equ-defSUD}
and $\hat{\mathbb{G}}_{m}(x):=m^{-1}\sum_{i=1}^{m} \ind{p_i\leq x}$ is the empirical
c.d.f. of the $p$-values.
\end{lemma}

We now state our main result.

\begin{theorem} \label{thm}
Consider a threshold collection $\mbf{t}$ of the form 
$t_k=\rho(k/m)$, where $\rho:[0,1]\rightarrow [0,1]$ satisfies \eqref{rhocond1} and  \eqref{rhocond2}. 
Let $\zeta\in (0,1)$ be an arbitrary constant.
For $\delta\in(0,1)$, define
\[
u^+_\delta = \TSud(\lambda,(G^{DU}_\zeta + \delta)\wedge 1)
\qquad \text{ and } \qquad
u^-_\delta= \TSud(\lambda,(G^{DU}_\zeta - \delta)\vee 0),
\]
where $G^{DU}_\zeta(x):=(1-\zeta) + \zeta x$
(see Figure~\ref{fig:finitesample} for an illustration).
Let us define the remaining term: for any $y\in(0,1)$,
  \begin{align}
 \eps(\delta,m,\zeta,y):=\frac{\rho(u^+_\delta)-\rho(u^-_\delta)}{u^+_\delta} + \frac{4}{1-\zeta} e^{-2 m \left(\delta-y-\frac{1}{m}\right)_+^2 (1-y/\zeta)_+ }  \label{equ-remainingterm}.
  \end{align}
Then, for any $F\in\mathcal{F}$ and $\lambda\in\{1,\dots,m\}$ the following holds.
\begin{itemize}
\item[\textbullet] In the $\FM$ model with $0<m_0<m$ and $\nu=\max_{k\in\{m_0-1,m_0\}}\{|k/m-\zeta|\}\in [0,1]$, we have
  \begin{align}
 \FDR(\SUD(\mbf{t}),m_0,F) \leq& \:\FDR(\SUD(\mbf{t}),m_0,F\equiv1) + \frac{m_0}{m} \eps(\delta,m,\zeta,\nu); \label{equ-final}
  \end{align}
\item[\textbullet] In the $\RM$ model with $\pi_0=\zeta$, we have for any $\gamma\in(0,1)$,
  \begin{align}
 \FDR(\SUD(\mbf{t}),\pi_0,F) \leq& \:\FDR(\SUD(\mbf{t}),\pi_0,F\equiv1) + \pi_0 \eps(\delta,m,\zeta,\gamma) + 4 e^{-2 m (\gamma-1/m)_+^2}. \label{equ-final-RM}
  \end{align}
 \end{itemize}
\end{theorem}

 \begin{figure}[htbp]
\begin{center}
\begin{tabular}{cc}
Linear rejection curve ($\alpha=0.5$)& AORC  ($\alpha=0.2$)\\
\includegraphics[scale=0.35]{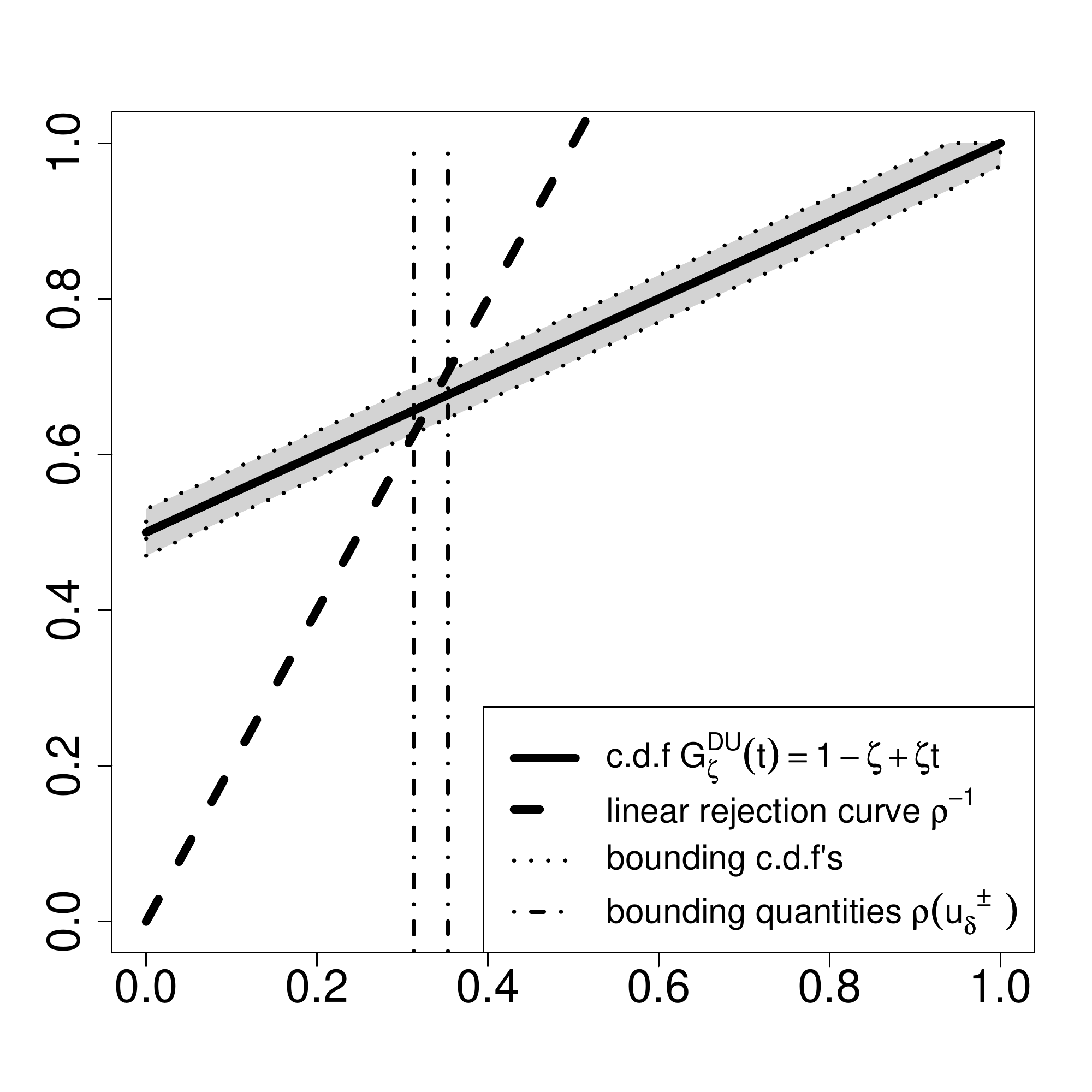}&\includegraphics[scale=0.35]{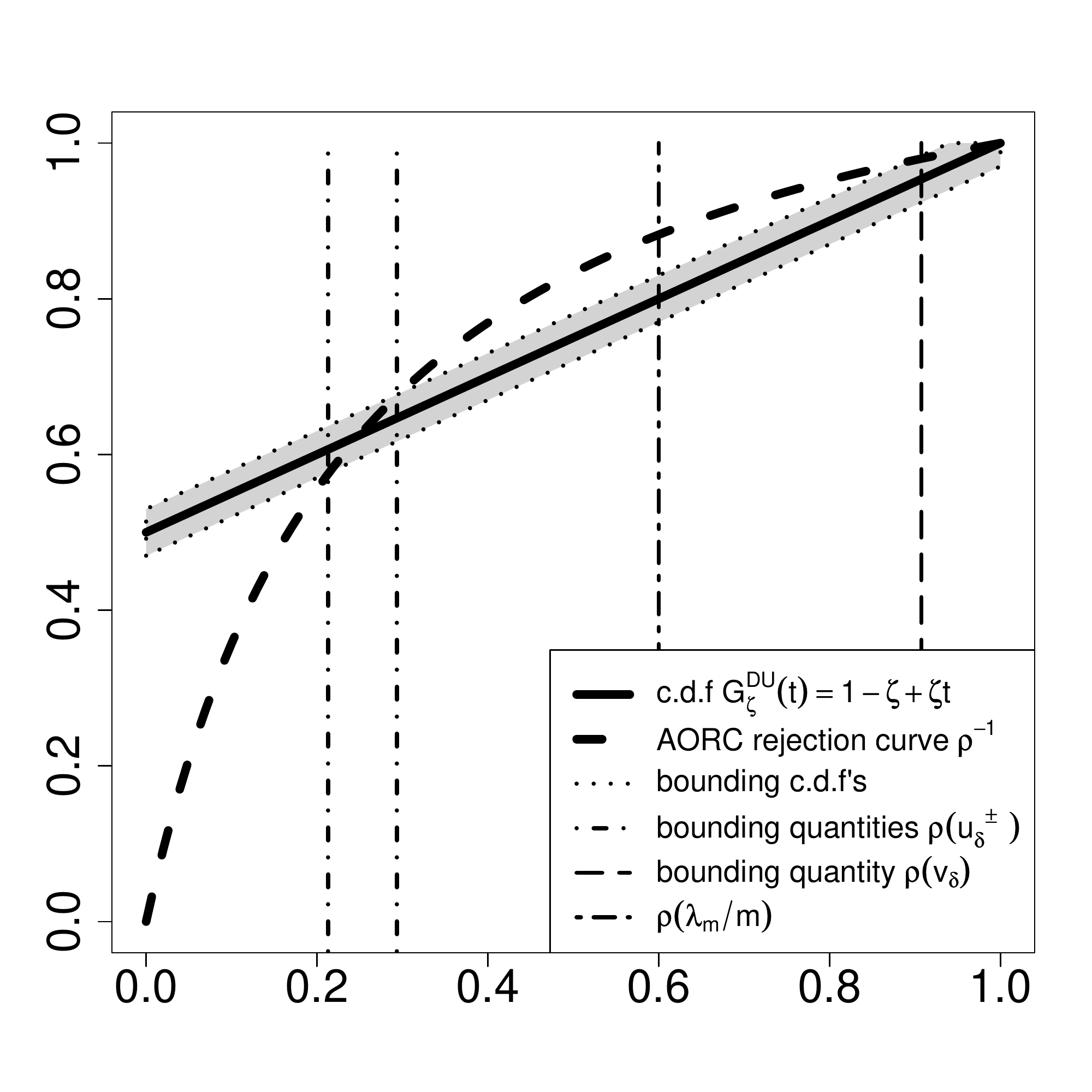}
\end{tabular}
\caption{Illustration for $u^{+}_\delta$ and $u^-_\delta$ in the case of the LSUD and the SUD based on AORC. Here the $X$-axis is on the ``threshold scale" $t=\rho(u)$. $\zeta=0.5$; $\delta=0.03$. The area between $(G^{DU}_\zeta - \delta)\vee 0$ and $(G^{DU}_\zeta + \delta)\wedge 1$ is displayed in gray.}
\label{fig:finitesample}
\end{center}
\end{figure}

Theorem~\ref{thm} is proved in Section~\ref{sec:proof-mainthm}. This result is illustrated in the two following examples.
\begin{enumerate}
\item Let us first apply Theorem~\ref{thm} in the particular case where $\rho(x)=\alpha x$ is the linear critical value function, see the left panel of Figure~\ref{fig:finitesample}. In that case, we have  $u^-_\delta= \frac{1-\zeta-\delta}{1-\alpha\zeta}\vee 0$ and $u^+_\delta= \frac{1-\zeta+\delta}{1-\alpha\zeta} \wedge 1$. Hence, $$(\rho(u^+_\delta)-\rho(u^-_\delta))/u^+_\delta \leq \frac{2\alpha\delta}{1-\zeta+\delta}.$$
As a result, \eqref{equ-final} and \eqref{equ-final-RM} hold by replacing $(\rho(u^+_\delta)-\rho(u^-_\delta))/u^+_\delta$ by $\frac{2\alpha\delta}{1-\zeta+\delta}$ inside the remaining term $\eps(\delta,m,\zeta,y)$. 
\item Second, for $\zeta>\alpha$, let us use
 \begin{align}\label{rhoAORC}\rho(u)=\alpha u/(1-u(1-\alpha)) \mbox{, that is, }\rho^{-1}(t)=t/(\alpha+t(1-\alpha)) .
\end{align}
The rejection curve $\rho^{-1}$, displayed in the right panel of Figure~\ref{fig:finitesample}, is called the ``asymptotically optimal rejection curve"  (AORC). It was introduced by 
\citet{FDR2009} for the purpose of improving the power of the linear critical value function. When applying Theorem~\ref{thm} with this choice of $\rho$, the calculation of $u^-_\delta $ and $u^+_\delta $ depends on the position of the parameter $\lambda_m/m$ on $[0,1]$ and $\rho(u^+_\delta)-\rho(u^-_\delta)$ may not vanish when $\delta$ becomes small,  see  Figure~\ref{fig:finitesample}.
Fortunately, when $\rho(\lambda_m/m)$ (dotted-long-dashed line) is smaller than the quantity $\rho(v_\delta)$  (dashed line), the two points $\rho(u^-_\delta)$ and $\rho(u^+_\delta)$ (dotted-dashed lines) are expected to be close as $\delta$ becomes small and the bound given in  Theorem~\ref{thm} will vanish.
The exact expressions of $u^-_\delta $, $u^+_\delta $ and $v_\delta$ can be easily derived by solving the corresponding quadratic equations. For short, we only report the equivalent as $\delta$ tends to $0$:
\begin{align*}
v_\delta&=1-\delta\alpha/(\zeta-\alpha)+O(\delta^2) \\
u^+_\delta& =(1-\zeta)/(1-\alpha) + \delta \zeta/(\zeta-\alpha)+O(\delta^2)\\
u^-_\delta&=(1-\zeta)/(1-\alpha) - \delta \zeta/(\zeta-\alpha)+O(\delta^2).
\end{align*}
Since $ \rho'((1-\zeta)/(1-\alpha)) = \alpha/\zeta^2$, we have $\rho(u^+_\delta)-\rho(u^-_\delta) = 2\alpha \delta / (\zeta^2-\alpha\zeta)+O(\delta^2)$. 
As a result, assuming $\zeta>\alpha$ and $\lambda/m<v_\delta$, we can derive that \eqref{equ-final} and \eqref{equ-final-RM} hold and that quantity $ (\rho(u^+_\delta)-\rho(u^-_\delta))/u^+_\delta$ is equivalent to $ \frac{ 2\alpha(1-\alpha)}{(\zeta^2-\alpha\zeta)(1-\zeta)}\delta$ (as $\delta$ tends to zero)  in the remaining term $\eps(\delta,m,\zeta,y)$.

\end{enumerate}

 \subsection{Convergence rate when \texorpdfstring{$m\rightarrow\infty$}{m to infinity}}
 
We can now use Theorem~\ref{thm} in an asymptotic sense and for specific critical value functions,
in order to obtain an explicit bound on the  convergence rate of the limit appearing in \eqref{equ_asymp}.

 \begin{corollary} \label{thmasymp}
Let $\alpha\in(0,1)$. Consider a threshold collection of the form 
$t_k=\rho(k/m)$, where $\rho$ is either $\rho(x)=\alpha x$ (linear) or given by \eqref{rhoAORC}, that is, associated 
with the AORC.
Consider the SUD procedure of threshold collection $\mbf{t}$ and of order $\lambda=\lambda_m$ possibly depending on $m$. 
Consider either the $\FM$ model with $m_0=\lfloor \zeta_m m\rfloor$ or $\RM$ model with $\pi_0 = \zeta_m$, for some $\zeta_m\in (\alpha,1)$ possibly depending on $m$.
Assume that 
$\frac{1}{1-\zeta_m}\sqrt{(\log m)/m} = o(1)$. For the AORC case, assume moreover 
$\frac{1}{1-\lambda_m/m} \sqrt{(\log m)/m} = o(1)$.
Then we have  for any $F\in\mathcal{F}$,
\begin{equation}\label{equ_asymp_rate}
\left( \FDR(\SUDm(\mbf{t}),F) - \FDR(\SUDm(\mbf{t}),F\equiv 1) \right)_+= O\left(\frac{1}{1-\zeta_m}\sqrt{\frac{\log m}{m}}\right).\end{equation}
\end{corollary}

Corollary~\ref{thmasymp} is proved in Section~\ref{proof:thmasymp} and is an easy consequence of Theorem~\ref{thm} when taking $\delta=\delta_m$ (and $\gamma=\gamma_m$) suitably tending to zero.
Assuming $\zeta_m>\alpha$ is not an important restriction because when $\zeta_m\leq \alpha$, controlling the FDR is a trivial problem: the procedure rejecting all the hypotheses has an FDR (and even an FDP) smaller than $\zeta_m\leq \alpha$.

While focusing on the linear and AORC rejection curve, the conclusion of Corollary~\ref{thmasymp} is substantially more informative than Theorem~\ref{th-Gont}: it evaluates what is the order of the error when considering that the DU is an LFC of an SUD test. For $\zeta_m=\zeta\in(0,1)$ fixed with $m$, note that the rate of convergence in \eqref{equ_asymp_rate} is equal to the parametric rate, up to a  $\log m$ factor. Furthermore, the constant in the $O(\cdot)$ can be derived explicitly by using the bound from the previous section. For $\zeta_m$ tending  to $1$ (not too quickly, ``fairly" sparse case), the convergence rate is slower.
 
As a counterpart,  assumptions of Corollary~\ref{thmasymp} are more restrictive than those of Theorem~\ref{th-Gont}. In particular, they exclude 
the case where 
$\zeta_m$ tends to $1$ faster than $\sqrt{(\log m)/m}$ (``highly" sparse case).
This is a limitation of the methodology used to prove
the nonasymptotic results. This problem may possibly be fixed by adapting the proof of Lemma~3.7 of \citet{Gont2010} to a nonasymptotic setting, but this falls outside of the intended scope of this paper.

\section{Exact formulas}\label{sec:exactformulas}

In this section, we gather some of the formulas derived by \citet{RV2010,RV2010sup} to calculate the joint distribution of the number of false discoveries and the number of discoveries. Moreover, we complement this work by giving a new recursion that makes these formulas fully usable.
These calculations are used to state Proposition~\ref{prop:disproof}.

\subsection{A new Steck's recursion}

For any $k\geq 0$ and any threshold collection $\mbf{t}= (t_1,\dots,t_k)$, we denote
\begin{equation}
\Psi_k(\mbf{t}) = \Psi_k(t_{1},\dots,t_{k}) =\prob{U_{(1)}\leq t_{1}, \dots, U_{(k)}\leq t_k}.\label{equ_psi}
\end{equation}
 where $(U_i)_{1\leq i\leq k}$ is a sequence of $k$ variables i.i.d. uniform on $(0,1)$ and
with the convention $\Psi_0(\cdot)=1$.
In practice, quantity \eqref{equ_psi} can be evaluated using  standard 
 Steck's recursion
$\Psi_k(\mbf{t})  = (t_k)^k- \sum_{j=0}^{k-2} { {k}\choose{j}} (t_k-t_{j+1})^{k-j} \Psi_{j}(t_1,\dots,t_j)$
  \citep[p. 366--369]{SW1986}.

Next, we generalize the latter to the case of two populations. Define for $0\leq k_0\leq k$  and any threshold 
collection $\mbf{t}= (t_1,\dots,t_k)$,
\begin{equation}
\Psi_{k,k_0,F}(t_1,\dots,t_k)=\prob{U_{(1)}\leq t_{1}, \dots, U_{(k)}\leq t_k},\label{equ_psi2}
\end{equation}
 where $(U_i)_{1\leq i\leq k}$ is a sequence of $k$ variables such that $(U_i)_{1\leq i\leq k_0}$ are i.i.d. uniform on $(0,1)$, independently of $(U_i)_{k_0+1\leq i\leq k}$ i.i.d. of c.d.f. $F$ and with the convention $\Psi_{0,0,F}(\cdot)=1$.
The computation of $\Psi_{k,k_0,F}$ is more difficult than $\Psi_{k}$ because it involves non i.i.d. variables. To our knowledge the existing formulas for computing $\Psi_{k,k_0,F}$ have a complexity exponential with $k$ \citep{glueck2008-2}. Here, we propose a substantially less complex computation, by generalizing Steck's recursions as follows.
 
\begin{proposition}\label{prop:newrecursions}
The following recursion holds: for $0\leq k_0\leq k$,
\begin{align}
\Psi_{k,k_0,F}(t_1,\dots,t_k)  =& (t_k)^{k_0} F(t_k)^{k-k_0}- \sum_{\substack{0\leq j_0\leq j\leq k-2\\ j_0\leq k_0\\ j-j_0\leq k-k_0}}  { {k_0}\choose{j_0}} { {k-k_0}\choose{j-j_0}}\nonumber\\
& \times (t_k-t_{j+1})^{k_0-j_0}(F(t_k)-F(t_{j+1}))^{k-k_0-j+j_0}  \Psi_{j,j_0,F}(t_1,\dots,t_j)\label{gen-steck},
\end{align}
with the convention $0^0=1$.
\end{proposition}

This is proved  in Section~\ref{proof:newrecursions}. Note that the case $k=k_0$ reduces to the standard (one population) Steck's recursion.

\subsection{FDR formulas}

Using the $\Psi_k$'s and $\Psi_{k,k_0}$'s, let us define the following useful quantities: for any threshold 
collection $\mbf{t}=(t_k)_{1\leq k \leq m}$, $F\in\mathcal{F}$; for any $\pi_0\in[0,1]$, $k\geq 0$, $k\leq m$, $0\leq j\leq k$,  we let
\begin{align}
\surm_{m,\pi_0,F}(\mbf{t},k,j) &= {{m} \choose {j}} {{m-j} \choose {k-j}} \pi_0^j\pi_1^{k-j} (t_k)^j (F(t_k))^{k-j}\nonumber\\
&\times\Psi_{m-k}(1-G(t_{m}),\dots,1-G(t_{k+1})) \label{equ_for_jointdistsu_rm};\\
\sdrm_{m,\pi_0,F}(\mbf{t},k,j) &= {{m} \choose {j}} {{m-j} \choose {k-j}} \pi_0^j\pi_1^{k-j}  (1-G(t_{k+1}))^{m-k} \nonumber\\
&\times \Psi_{k,j,{F}}(t_1,\dots,t_{k})\label{equ_for_jointdistsd_rm},
\end{align}
where $G(t)=\pi_0 t + (1-\pi_0) F(t)$. For any $m_0\in\{0,\dots,m\}$, $k\geq 0$, $k\leq m$, $j\leq m_0$, $k-j\leq m-m_0$, we let
\begin{align}
\sufm_{m,m_0,F}(\mbf{t},k,j) &= {{m_0} \choose {j}} {{m-m_0} \choose {k-j}} (t_k)^j (F(t_k))^{k-j}\nonumber\\
&\times\Psi_{m-k,m_0-j,\ol{F}}(1-t_{m},\dots,1-t_{k+1}) \label{equ_for_jointdistsu_fm};\\
\sdfm_{m,m_0,F}(\mbf{t},k,j) &= {{m_0} \choose {j}} {{m-m_0} \choose {k-j}} (1-t_{k+1})^{m_0-j} (1-F(t_{k+1}))^{m-m_0-k+j}\nonumber\\
&\times\Psi_{k,j,{F}}(t_1,\dots,t_{k}),\label{equ_for_jointdistsd_fm}
\end{align}
where  $\ol{F}(t)=1-F(1-t)$. The following results have been proved by \citet{RV2010,RV2010sup}.

\begin{theorem}[Roquain and Villers (2011)]\label{th-RV}
Consider any threshold collection $\mbf{t}$ and the quantities defined by \eqref{equ_for_jointdistsu_rm}, \eqref{equ_for_jointdistsd_rm}, \eqref{equ_for_jointdistsu_fm} and \eqref{equ_for_jointdistsd_fm}.
Then the following holds:
\begin{itemize}
\item[(i)] In the $\RM$ model, for any $\pi_0\in[0,1]$, $F\in\mathcal{F}$, $0 \leq k \leq m$, $0\leq j\leq k$, 
\begin{align}
\P(|R\cap \{1,\dots,m_0\}| = j, |R| = k )= \left\{\begin{array}{cc} \surm_{m,\pi_0,F}(\mbf{t},k,j) & \mbox{ for $R=\SU(\mbf{t})$,}\\\sdrm_{m,\pi_0,F}(\mbf{t},k,j) & \mbox{ for $R=\SD(\mbf{t})$.} \end{array} \right.
\label{proba-RM}
\end{align}
\item[(ii)]
In the $\FM$ model, for any $m_0\in\{0,\dots,m\}$,  $F\in\mathcal{F}$, $0 \leq k \leq m$, $0\vee (k-m+m_0) \leq j\leq m_0\wedge k$, 
\begin{align}
\P(|R\cap \{1,\dots,m_0\}| = j, |R| = k )= \left\{\begin{array}{cc} \sufm_{m,m_0,F}(\mbf{t},k,j) & \mbox{ for $R=\SU(\mbf{t})$,}\\\sdfm_{m,m_0,F}(\mbf{t},k,j) & \mbox{ for $R=\SD(\mbf{t})$.} \end{array} \right.
\label{proba-FM}
\end{align}
\end{itemize}
\end{theorem}

Classically, any step-up-down procedure can be written as a combination of a step-down and a step-up procedure \citep{Sar2002}:
 \begin{equation}\label{SUD-comb}
 \SUD(\mbf{t})= \left\{\begin{array}{ll} 
  \SU((t_\lambda\wedge t_j)_{1\leq j \leq m})& \mbox{ if $|\SU((t_\lambda\wedge t_j)_{1\leq j \leq m})|< \lambda$,}\\
 \SD((t_\lambda\vee t_j)_{1\leq j \leq m})& \mbox{ if $| \SD((t_\lambda\vee t_j)_{1\leq j \leq m})|\geq \lambda$.}
\end{array}\right.
\end{equation}
Moreover, since $\{|\SU((t_\lambda\wedge t_j)_{1\leq j \leq m})|< \lambda\}=\{p_{(\lambda)}>t_\lambda\}$ and $\{| \SD((t_\lambda\vee t_j)_{1\leq j \leq m})|\geq \lambda\}=\{p_{(\lambda)}\leq t_\lambda\}$ the two cases in \eqref{SUD-comb} form a partition of the probability space. 
This yields the following explicit FDR computations: 

\begin{corollary}
let $\lambda\in\{1,\dots,m\}$ and consider any threshold collection $\mbf{t}$. Then the following holds:
\begin{itemize}
\item[(i)] In the  model $\RM$, for any $\pi_0\in[0,1]$,  $F\in\mathcal{F}$, we have
\begin{align}
\FDR(\SUD(\mbf{t}))=& \sum_{k=1}^{\lambda-1} \sum_{j=0}^{k} \frac{j}{k}  \surm_{m,\pi_0,F}(\mbf{t}\wedge t_\lambda,k,j)  \nonumber\\
&+\sum_{k=\lambda}^m \sum_{j=0}^{k} \frac{j}{k}  \sdrm_{m,\pi_0,F}(\mbf{t}\vee t_\lambda,k,j).
\label{form-SUD-RM}
\end{align}
\item[(ii)] In the  model $\FM$, for any $m_0\in\{0,\dots,m\}$,  $F\in\mathcal{F}$, we have
\begin{align}
\FDR(\SUD(\mbf{t}))=& \sum_{k=1}^{\lambda-1} \sum_{j=0\vee (k-m+m_0)}^{m_0\wedge k} \frac{j}{k}  \sufm_{m,m_0,F}(\mbf{t}\wedge t_\lambda,k,j)  \nonumber\\
&+\sum_{k=\lambda}^m \sum_{j=0\vee (k-m+m_0)}^{m_0\wedge k} \frac{j}{k}  \sdfm_{m,m_0,F}(\mbf{t}\vee t_\lambda,k,j).\label{form-SUD-FM}
\end{align}
\end{itemize}
\end{corollary}

In the model $\RM$, it turns out that $\FDR(\SUD(\mbf{t}))$ has an expression only involving the $\Psi_k$s, and not the $\Psi_{k,j,{F}}$ \citep[Section~5.2]{RV2010}. Although it has a somewhat less intuitive form, it is better than \eqref{form-SUD-RM} from a computational point of view. 

\section{Discussion}

Our aim in this paper was to address the question ``is the Dirac-uniform distribution an LFC for an intermediate step-up-down procedure (that uses a standard threshold collection)?" 
In a nutshell, the answer we found is ``no, but almost".
We provided a rigorous quantification 
of what ``almost" means, using 
an alternative approach to the asymptotic results of \citet{Gont2010} 
that entails nonasymptotic bounds and explicit convergence rates.
In practical situations, evaluating such bounds can allow to determine whether we can consider that the FDR is maximum when the signal strength is maximum.

Returning to equations \eqref{equ_FDR_FM} and \eqref{equ_FDP_FM}, an additional question,
particularly relevant in practice, 
is how appropriate it is to base the multiple type I error criterion solely
on control of the {\em expectation} of the random variable FDP. 
We notice that Theorem \ref{th-RV} may also be used to study this issue
by computing exactly the point mass function of the FDP under arbitrary configurations
for the alternative, cf. Section \ref{fdp-formulas} in the appendix.
Based on this, we investigated
to what extent the distribution of the FDP concentrates around its expectation for a simple
Gaussian location model with parameter $\mu$. 
Figure \ref{fig:FDPconc} was obtained from these exact formulas for
$m = 100$ and varying values of $\pi_0$ and $\mu$. 
Note that the unrealistically large choice of $\alpha = 1/2$
has only been used for reasons of readability of the figures;
similar plots also obtained when choosing $\alpha$ smaller
(the variance of the FDP actually increases with smaller $\alpha$, because this entails  a smaller
number of rejections). 
On inspection of these graphs, it becomes apparent 
that -- even though joint independence of the $p$-values holds -- 
the distribution of the FDP is not concentrated around the corresponding FDR in the 
following two situations: (i) The effect size $\mu$ is close to zero (weak signal) or (ii) 
the proportion $\pi_0$ of true null hypotheses is close to $1$ (sparse signal). 
Thus, controlling the FDR alone does not guarantee a small FDP for a specific experiment at hand
in these cases. 
For
 a well-defined dependency structure induced by exchangeable test statistics,
theoretical arguments for $m$ tending to infinity support
the observation that the distribution of the FDP often does not degenerate in the limit, see \citep{FDR2007,DR2011}.
For the jointly independent case and in the cases (i) or (ii) above, 
this phenomenon has not been theoretically studied 
to the best of our knowledge.
The latter can possibly be investigated by extending the asymptotic approach of \citet{Neu2008} to the case where $\mu$ and $\pi_0$ are allowed to depend on $m$.

\begin{figure}[h!]
\begin{tabular}{cccc}
&$\mu=0.01$& $\mu=5/\sqrt{m}=0.5$&$\mu=\infty$\\  
\rotatebox{90}{\hspace{1cm} $\pi_0=0.2$ 
} & \includegraphics[scale=0.23]{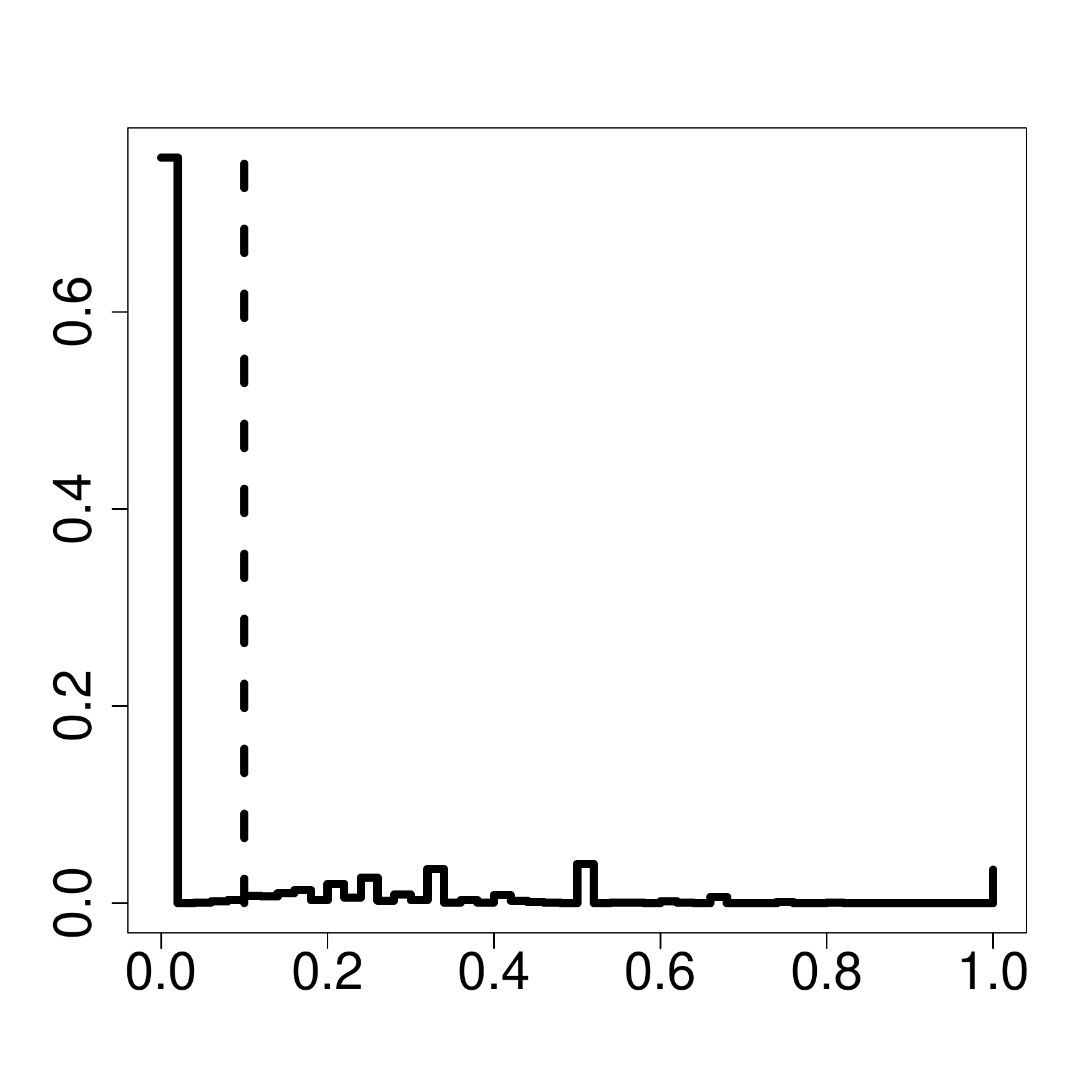} & \includegraphics[scale=0.23]{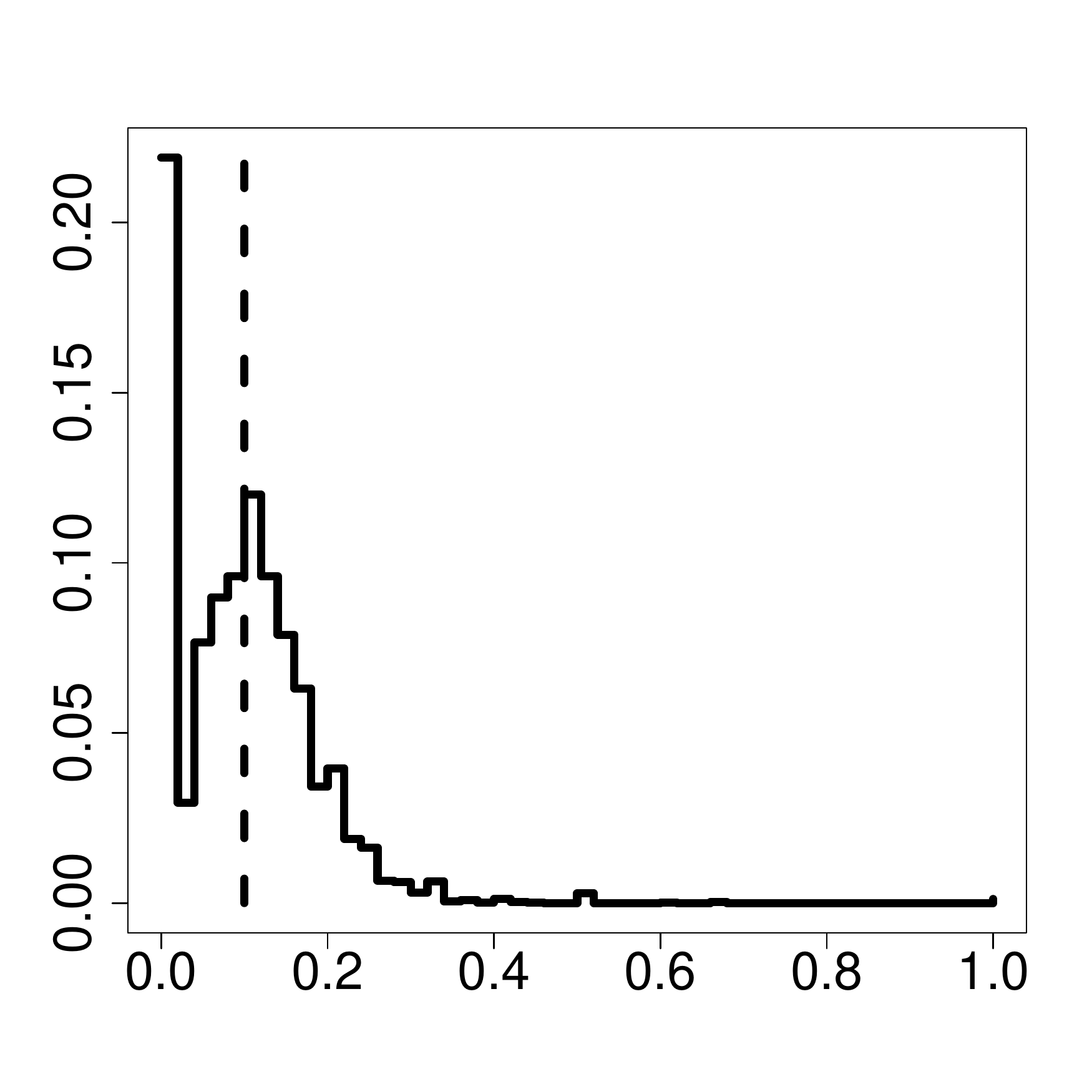} & \includegraphics[scale=0.23]{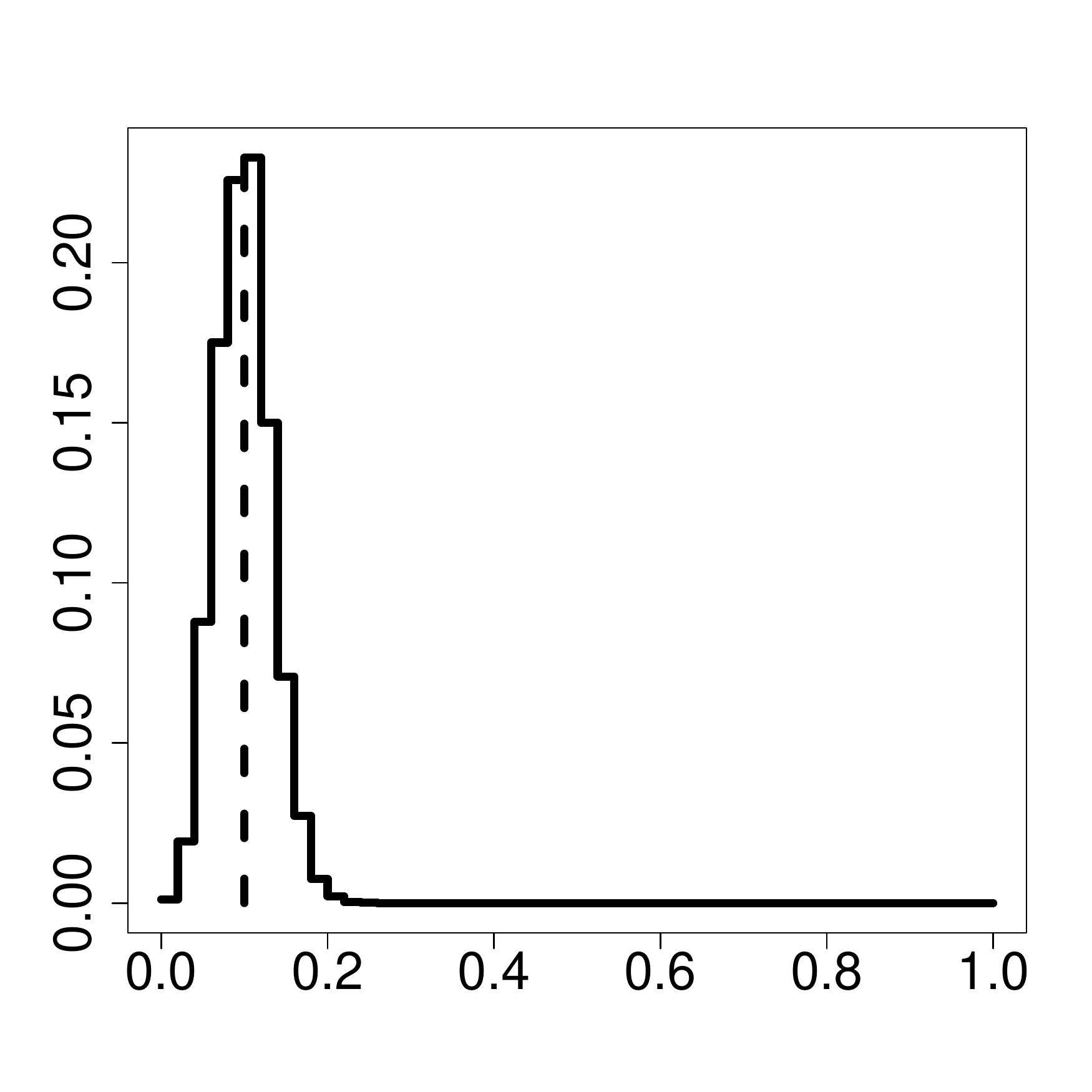} \\
\rotatebox{90}{\hspace{1cm}$\pi_0=0.5$  
} & \includegraphics[scale=0.23]{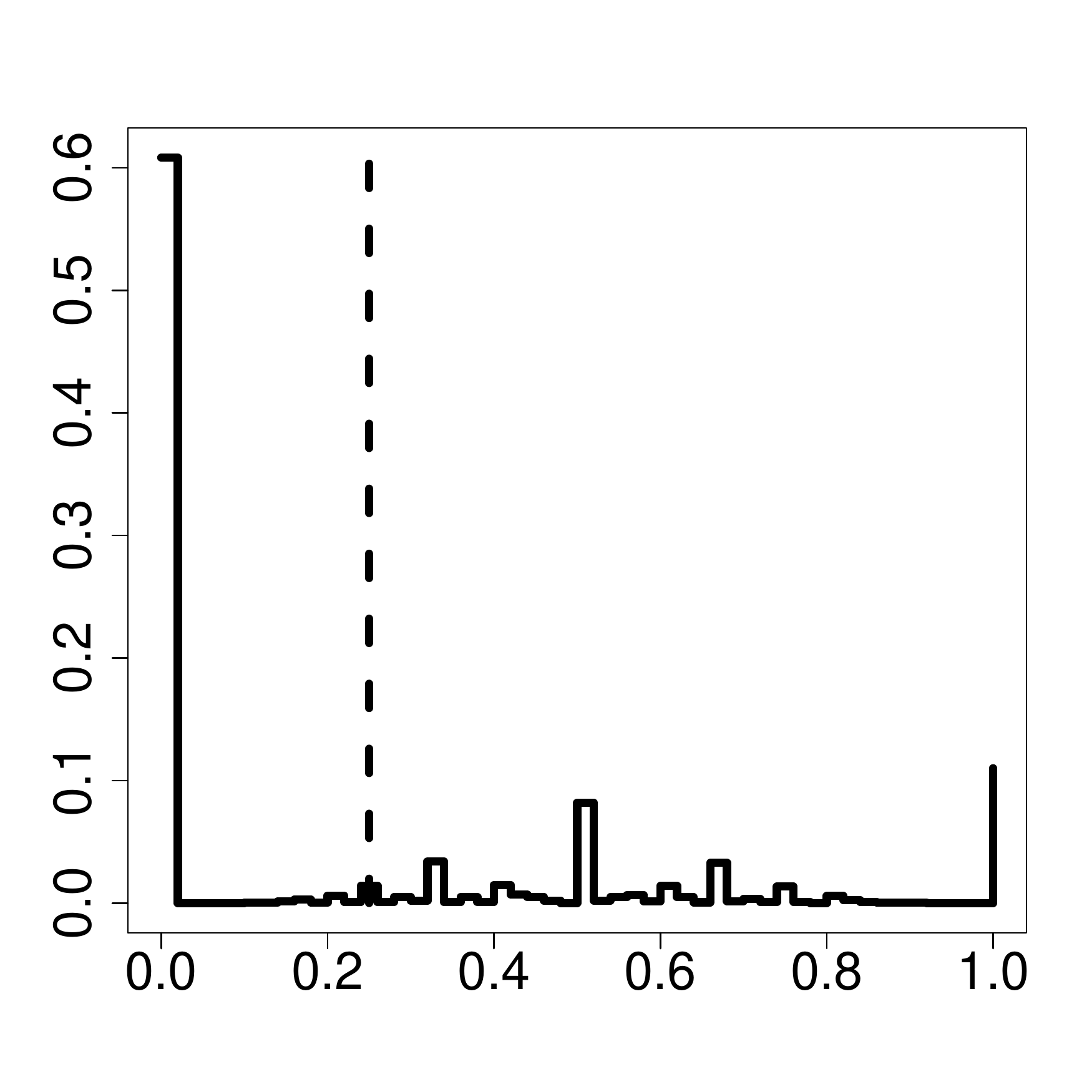} & \includegraphics[scale=0.23]{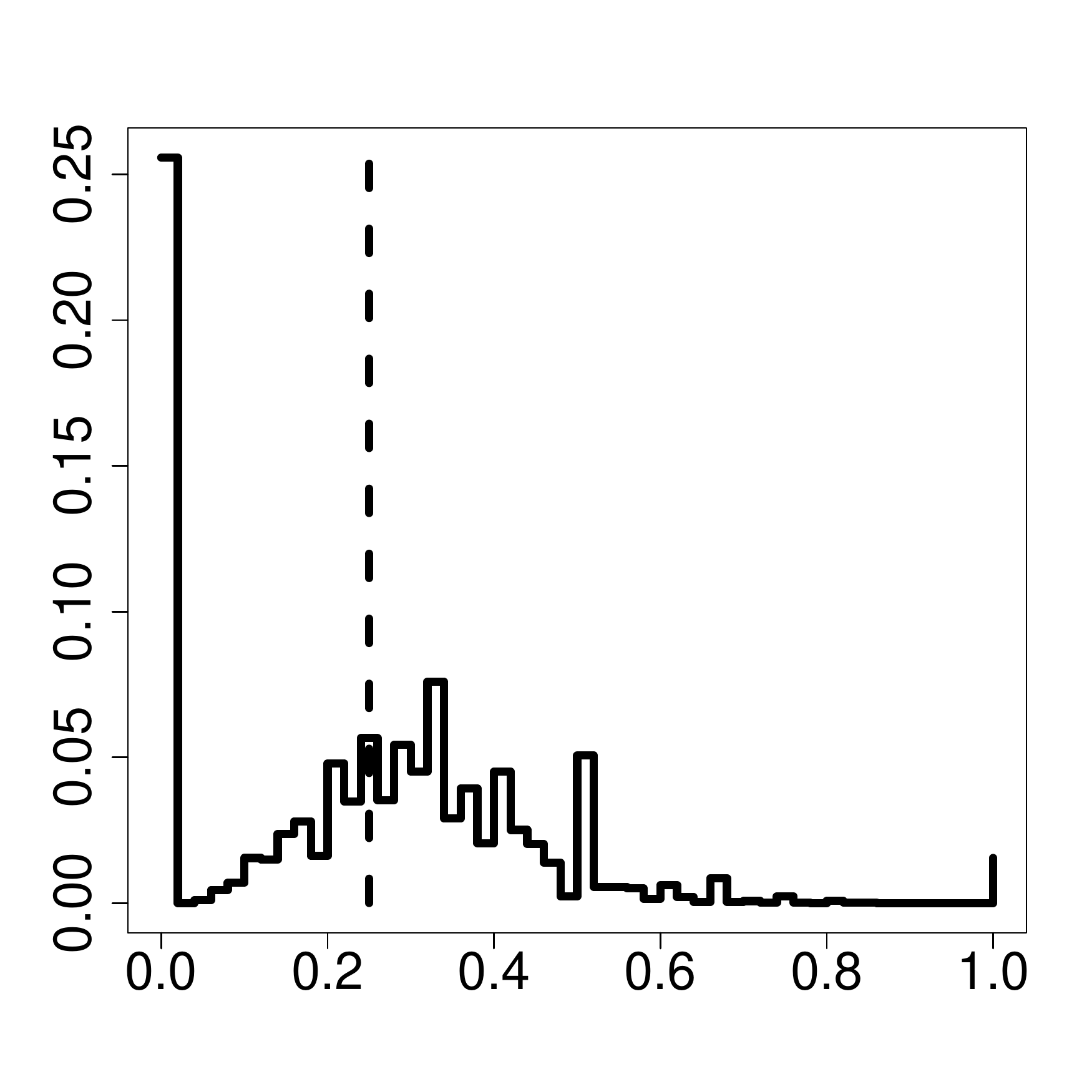} & \includegraphics[scale=0.23]{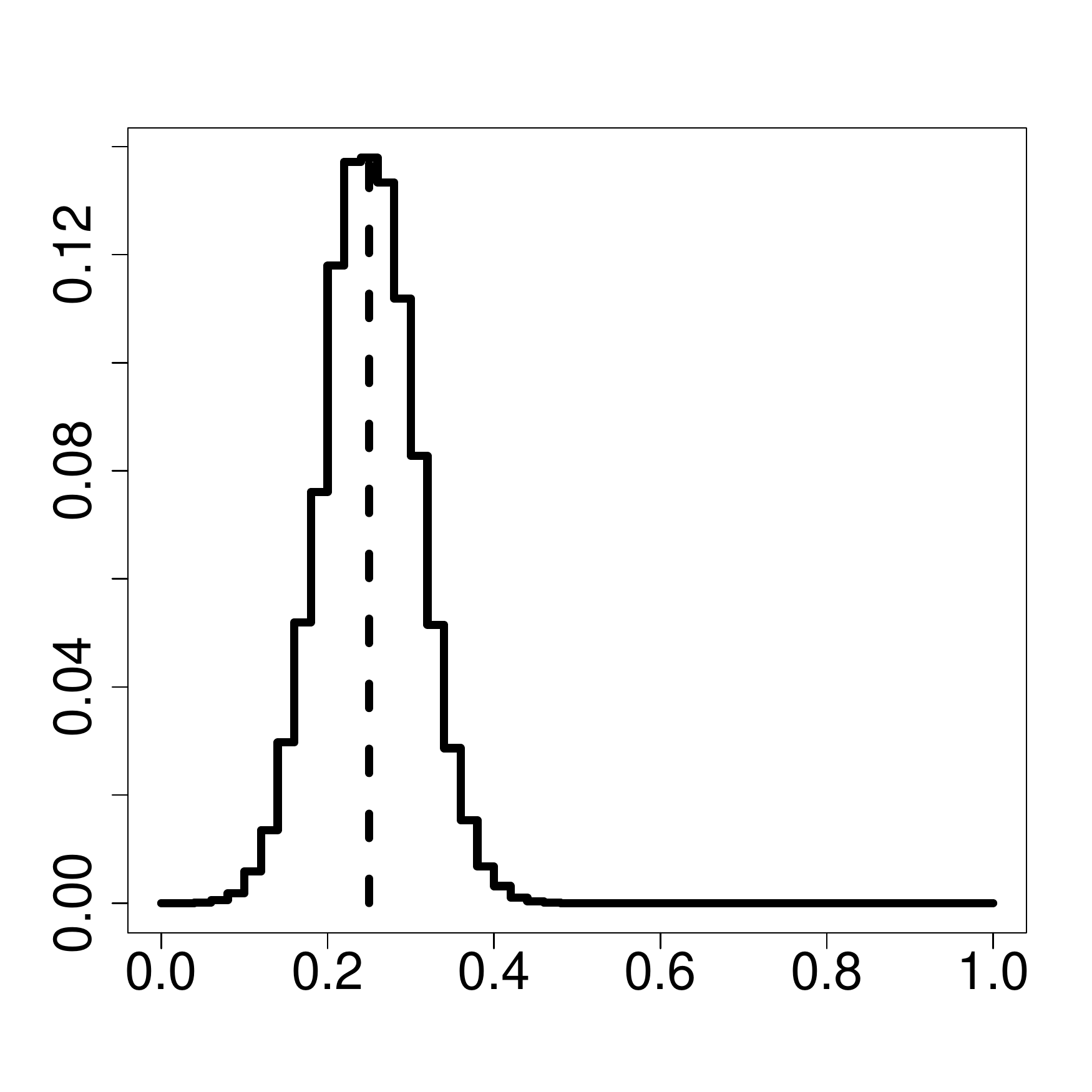} \\
\rotatebox{90}{\hspace{1cm}$\pi_0=0.95$  
} & \includegraphics[scale=0.23]{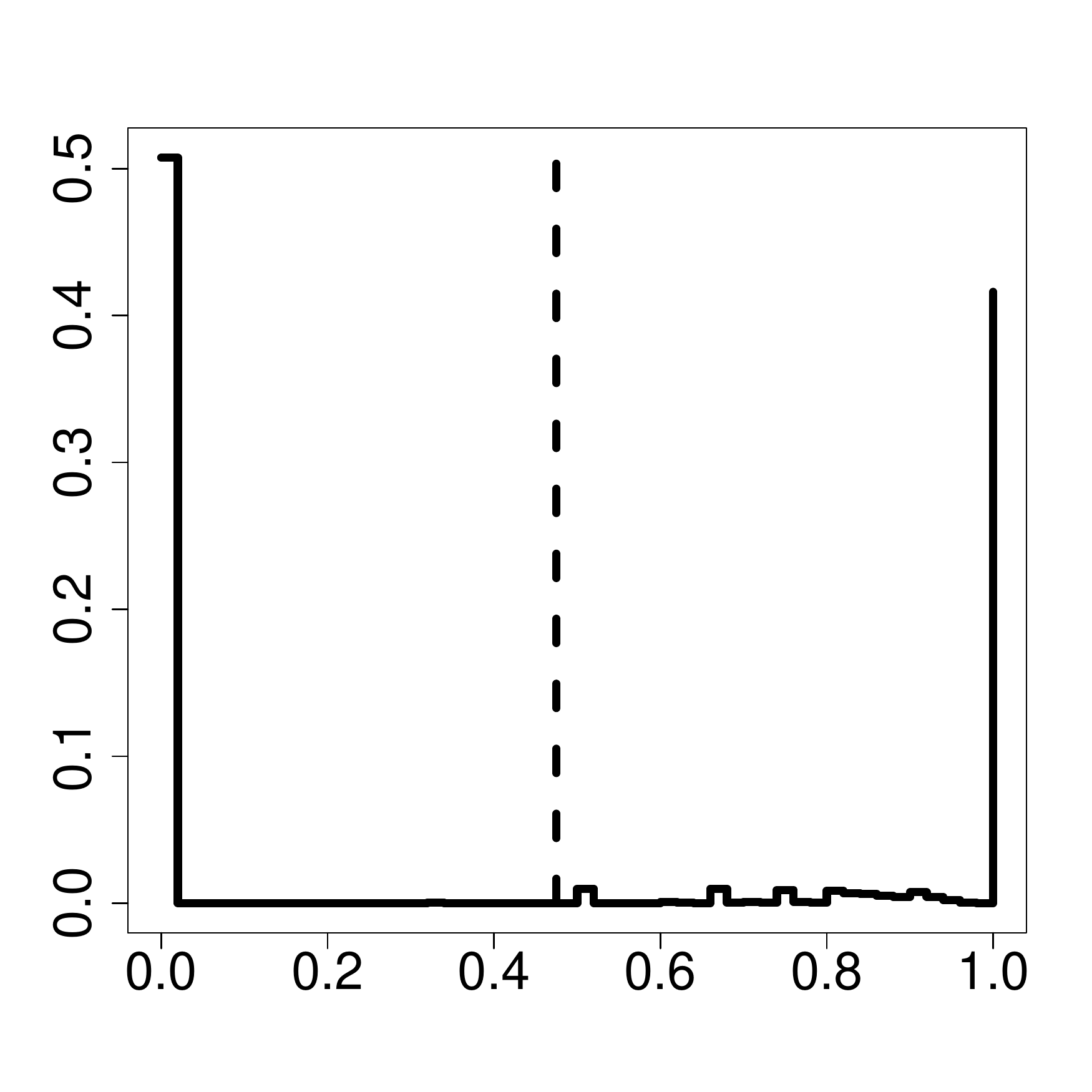} & \includegraphics[scale=0.23]{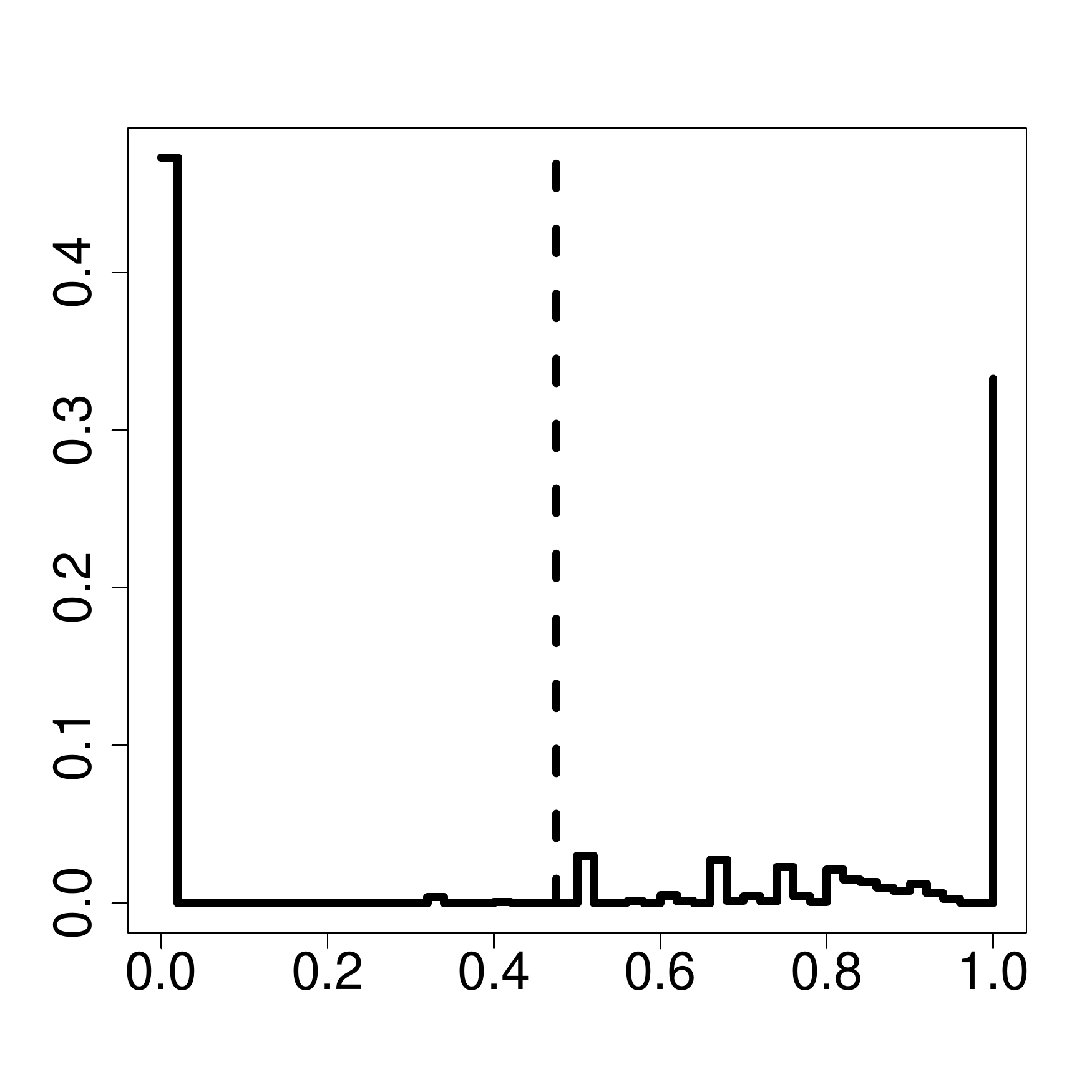} & \includegraphics[scale=0.23]{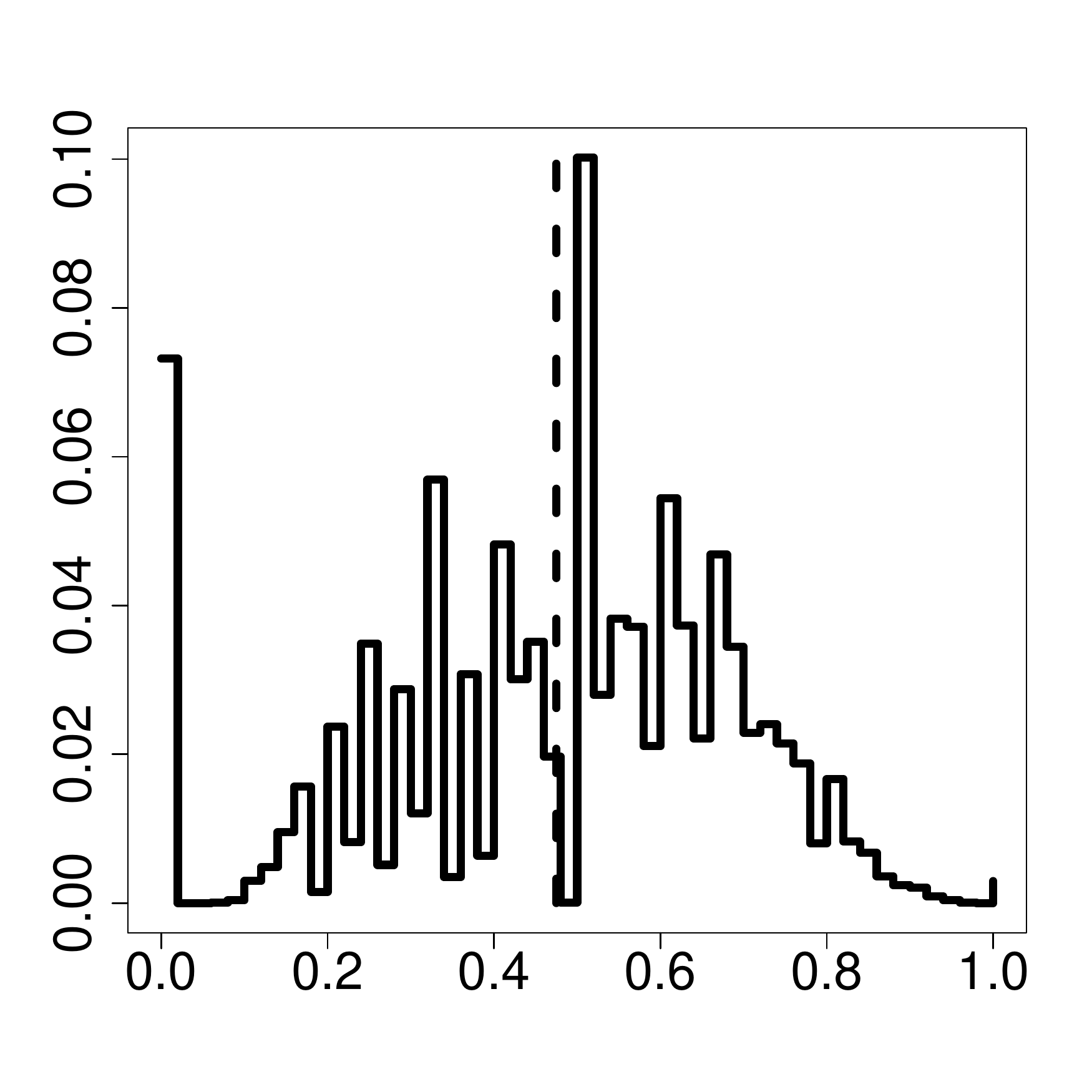} 
\end{tabular}
\caption{
Exact probability $\P(\FDP(\LSU)\in [i/50,(i+1)/50))$ for $0\leq i \leq 50$. The value of $\FDR(\LSU)=\pi_0\alpha$ is displayed by the vertical  dashed line.  Random mixture model. $\alpha=0.5$; $m=100$. One sided location Gaussian c.d.f. with parameter $\mu$. \label{fig:FDPconc}}
\end{figure}

Taking these considerations into account, control of the false discovery exceedance (i.e., of the probability that the FDP exceeds a given threshold) has recently been proposed in the 
literature \citep[see, e.g.][for a review]{farcomeni2008}. Controlling the false discovery exceedance control again brings forth the question of the corresponding LFC:
are Dirac-uniform configurations least favorable for, e.g., $\P(\FDP(\LSU)>x)$?
Some non-reported figures show
that this is not the case for any $x$. Hence, finding LFCs for the  false discovery exceedance stays 
an open avenue for future research.

\section{Proofs}\label{sec:proofs}

\subsection{Proof of Proposition~\ref{prop:newrecursions}}\label{proof:newrecursions}

We follow the proof of the regular Steck's recursion \citep[p. 366--369]{SW1986}. By using the convention $U_{(0)}=t_0=0$ and by considering the smallest $j$ for which $U_{(j+1)}> t_{j+1}$, we can write
\begin{align*}
& \P(U_{(k)}\leq t_k)-\prob{U_{(1)}\leq t_{1}, \dots, U_{(k)}\leq t_k} \\
& =  \sum_{j=0}^{k-2} \P(\forall i\leq j, U_{(i)}\leq t_i,  U_{(j+1)}> t_{j+1}, U_{k}\leq t_k) \\
&=  \sum_{j=0}^{k-2} \: \sum_{X\subset \{1,\dots,k\}, |X|=j} \P(\forall i\leq j, U_{(i)}\leq t_i,  \forall i \notin X, t_{j+1} \leq U_i\leq t_k).
\end{align*}
Hence, if $U_{(i:X)}$ denotes the $i$-th smallest member of the set $\{U_i,i\in X\}$, we obtain
\begin{align*}
& \P(U_{(k)}\leq t_k)-\prob{U_{(1)}\leq t_{1}, \dots, U_{(k)}\leq t_k} \\
=&  \sum_{j=0}^{k-2} \: \sum_{X\subset \{1,\dots,k\}, |X|=j} \P(\forall i\leq j, U_{(i:X)}\leq t_i) \P( \forall i \notin X, t_{j+1} \leq U_i\leq t_k)\\
=&  \sum_{j=0}^{k-2} \sum_{j_0=0}^{j} \: \sum_{X\subset \{1,\dots,k\}, |X|=j}\:  \ind{ |X \cap \{1,\dots,k_0\}|=j_0} \Psi_{j,j_0,F}(t_1,\dots,t_j) \\
&\times \P( \forall i \notin X, t_{j+1} \leq U_i\leq t_k).
\end{align*}

\subsection{Proof of Lemma~\ref{lemma-fund}}\label{proof-lemma-fund}
In this proof, we denote $\hat{\mathbb{G}}'_m=(\hat{\mathbb{G}}_m+m^{-1})\wedge 1$ for short.
Let us first note that for any $k\in\{1,\dots,m\}$, ``$p_{(k)}\leq t_k$" is equivalent to ``$\hat{\mathbb{G}}_{m}(\rho(k/m))\geq k/m$". 
We now distinguish two cases:
\begin{itemize}
\item[-] Step-up case:  assume $p_{(\lambda)} > t_\lambda$, that is, $\hat{\mathbb{G}}_{m}(\rho(\lambda/m))< \lambda/m$. Let us prove that $ \wh{k}/m=\TSud(\lambda/m,\hat{\mathbb{G}}_m)$. Since $\TSud(\toto,G)$ is a fixed point of the function $G\circ\rho$, we have $\TSud(\toto,G)\in\{0,1/m,\dots,m/m\}$. Hence, 
\begin{align*}
\widehat{k}/m&=  \max\{k/m\in\{0,\dots,\lambda/m\}\telque \:\hat{\mathbb{G}}_{m}(\rho(k/m))\geq k/m\} \\
&= \max\{u \in [0,\lambda/m] \telque \:\hat{\mathbb{G}}_{m}(\rho(u))\geq u\},
\end{align*}
and we can conclude.
\item[-] Step-down case: assume $p_{(\lambda)} \leq t_\lambda$, that is, $\hat{\mathbb{G}}_{m}(\rho(\lambda/m))\geq \lambda/m$. First assume that $\wh{k}<m$ and prove that $ (\wh{k}+1)/m=\TSud(\lambda/m,\hat{\mathbb{G}}'_m)$. 
On the one hand,
\begin{align}
(\widehat{k}+1)/m
&=\min\{k/m\in\{(\lambda+1)/m,\dots,m/m\}\telque \:\hat{\mathbb{G}}_{m}(\rho(k/m))< k/m \}\nonumber\\
&=\min\{k/m\in\{\lambda/m,\dots,m/m\}\telque \:\hat{\mathbb{G}}_{m}(\rho(k/m))< k/m \}\nonumber\\
&=\min\{k/m\in\{\lambda/m,\dots,m/m\}\telque \:\hat{\mathbb{G}}_{m}'(\rho(k/m)) \leq  k/m \},\label{equ-inter-1}
\end{align}
because $m\hat{\mathbb{G}}_{m}(\rho(k/m))$ is an integer.
On the other hand, since $\hat{\mathbb{G}}'_{m}(\rho(\lambda/m))\geq \lambda/m$, we have
\begin{align}
\TSud(\lambda/m,\hat{\mathbb{G}}'_m)&=\min \set{u \in [\lambda/m, 1] \telque \hat{\mathbb{G}}'_m(\rho(u))\leq u }\nonumber\\
&=\min \set{u \in \{\lambda/m,\dots,m/m\} \telque \hat{\mathbb{G}}'_m(\rho(u))\leq u }\label{equ-inter-2},
\end{align}
because $\TSud(\lambda/m,\hat{\mathbb{G}}'_m)\in\{0,1/m,\dots,m/m\}$.
Combining \eqref{equ-inter-1} and \eqref{equ-inter-2} yields the result. 
Second, in the case where $\wh{k}=m$, then for any $k/m\in\{\lambda/m,\dots,m/m\}$, we have $\hat{\mathbb{G}}_{m}(\rho(k/m))\geq k/m$. Hence, for all $k/m\in\{\lambda/m,\dots,(m-1)/m\}$, $\hat{\mathbb{G}}'_{m}(\rho(k/m))> k/m$ 
which entails $\TSud(\lambda/m,\hat{\mathbb{G}}'_m)=1$. Hence, $ \wh{k}/m=\TSud(\lambda/m,\hat{\mathbb{G}}'_m)$ in that case.
Finally, the inequality $\wh{k}/m\leq \TSud(\lambda/m,\hat{\mathbb{G}}'_m)$ always holds.
\end{itemize}

\begin{figure}[h!]
\begin{tabular}{ccc}
 $\lambda=8$ (SU part) & $\lambda=4$ (SD part) \\
  \includegraphics[scale=0.34]{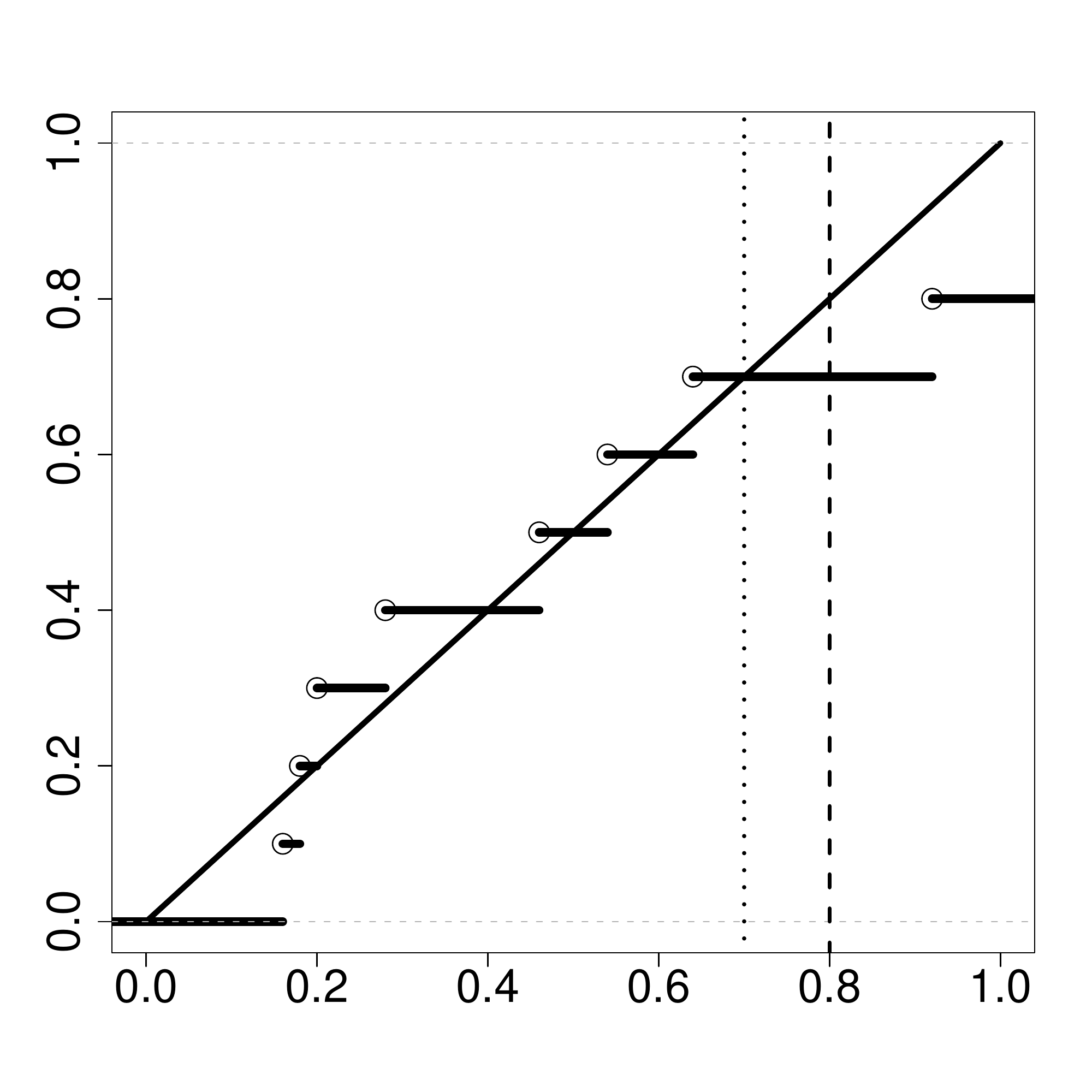}&\includegraphics[scale=0.34]{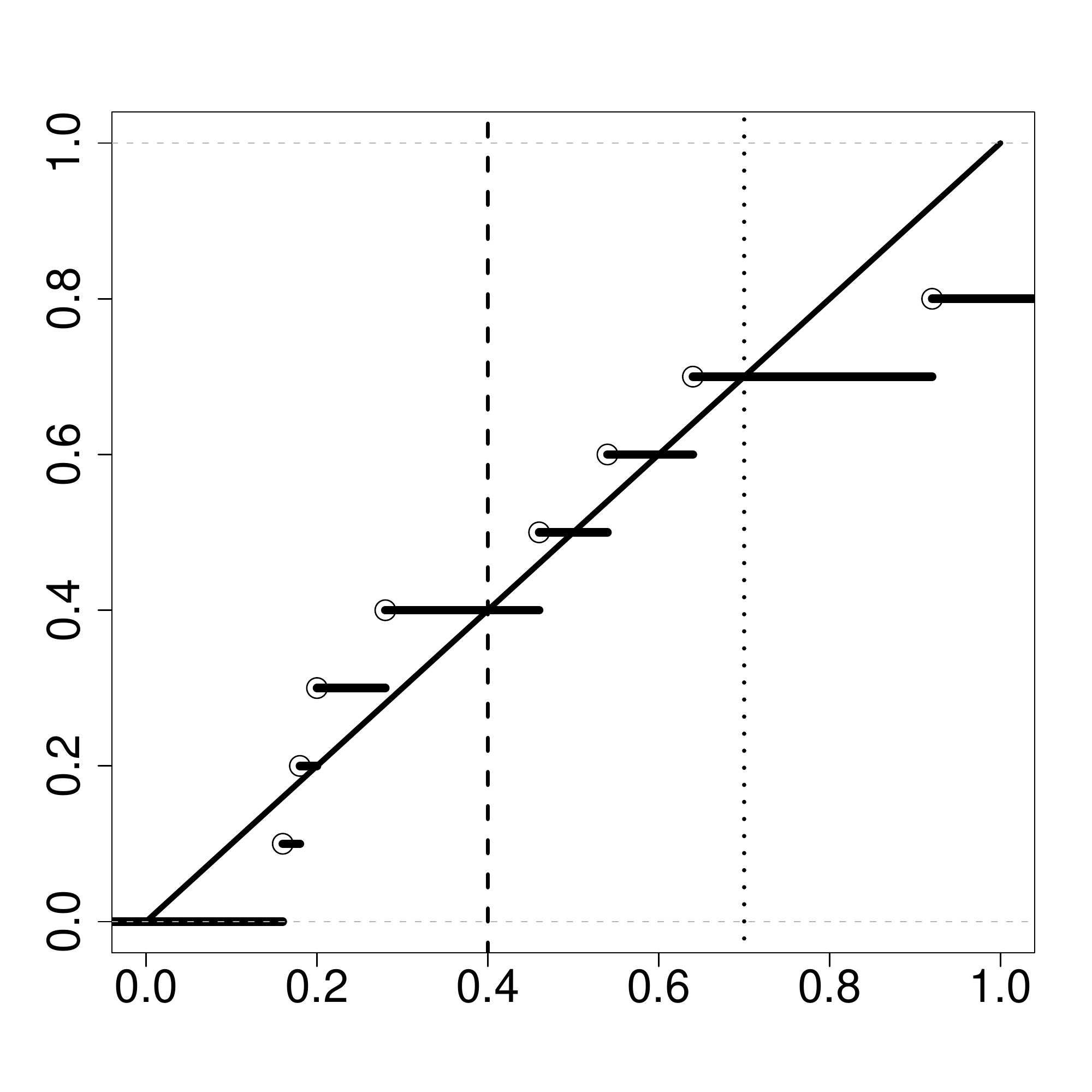}
  \end{tabular}
  \caption{Plot of $\mathbb{G}_m(\rho(\cdot))$ for $\rho(x)=0.5 x$ and $m=10$ $p$-values. The order $\lambda$
  of the SUD procedure is displayed by the dashed line while the value of $\wh{k}/m$ is displayed by the dotted line.
  Left: the SUD procedure of order $\lambda=8$ is such that $\wh{k}/m=\TSud(\lambda/m,\hat{\mathbb{G}}_m)=0.7$. Right: the SUD procedure of order $\lambda=4$ satisfies $\wh{k}/m=0.7$ but $\TSud(\lambda/m,\hat{\mathbb{G}}_m)=0.4$. \label{fig:contrex-SUD}}
\end{figure}

\subsection{Proof of Theorem~\ref{thm}}\label{sec:proof-mainthm}

Let us first prove the result in $\FM$ model. 
We recall $\hat{\mathbb{G}}_{m}(x):=m^{-1}\sum_{i=1}^{m} \ind{p_i\leq x}$, $G^{DU}_\zeta(x):=(1-\zeta) + \zeta x$ and we put $\hat{u}:=\hat{k}/m$, where $\hat{k}$ is defined by \eqref{equ-defSUD}. 
We can easily check from the definition \eqref{equ-Udef}
that if $G,G'$ are two nondecreasing functions such that
$\forall x \in [0,1], G(x) \geq G'(x)$, then $\TSud(\lambda,G) \geq \TSud(\lambda,G')$.
Based on the bound \eqref{equ-fund}, we deduce that $\forall \delta\in(0,1),$
\begin{align*}
\left\{\sup_{x\in[0,1]} |\hat{\mathbb{G}}_{m}(x)-G^{DU}_\zeta(x)| \leq  \delta -
\frac{1}{m} \right\}\subset \left\{u^-_\delta\leq \hat{u}\leq  u^+_\delta\right\}.
\end{align*}
As a consequence, in the $\mathrm{DU}(m,k)$ model, $k\in\{m_0-1,m_0\}$, by using $\hat{\mathbb{G}}_{m}(x)-G^{DU}_\zeta(x)= (k/m-\zeta) (1-\hat{\mathbb{G}}_{k}(x)) - \zeta (\hat{\mathbb{G}}_{k}(x)-x)$, we obtain
\begin{align}
\Omega_\delta(k)&:=\left\{\sup_{x\in[0,1]}|\hat{\mathbb{G}}_{k}(x)-x|   \leq   \zeta^{-1}(\delta -\nu -\frac{1}{m}) \right\} \subset \left\{u^-_\delta\leq \hat{u}\leq  u^+_\delta\right\}. \label{equ-borninfsup}
\end{align}
Remember that $p_1,\dots,p_{m_0}$ correspond to true nulls, hence,  when $k\in\{m_0-1,m_0\}$, $\hat{\mathbb{G}}_{k}$ involves only  variables which are i.i.d. uniform.
As a consequence, by using the DKW inequality with Massart's \citeyearpar{Mass1990} optimal constant, 
we have in the $\mathrm{DU}(m,k)$ model and for $k\in\{m_0-1,m_0\}$,
\begin{align}
\P_{DU(m,k)}[(\Omega_\delta(k))^c] &\leq 2\exp\left\{-2 k \left(\delta-\nu
-\frac{1}{m} \right)_+^2/\zeta^2 \right\}\nonumber\\
&\leq2\exp\left\{-2 m \left(\delta-\nu -\frac{1}{m} \right)_+^2 (1-\nu/\zeta)_+ \right\}\label{ineg-conc},
\end{align}
because $k/m\geq \zeta- \nu$ and $\zeta\leq 1$.

\subsubsection{Upper bound}

Let $q(x)=\rho(x)/x$ when $x\in (0,1]$ and $q(0)=\lim_{x\rightarrow 0^+} \rho(x)/x$ (the limit exists in $\R$ because $x\in(0,1]\rightarrow \rho(x)/x$ is non-decreasing).
Applying Theorem~4.3 of \citet{FDR2009}, we have 
\begin{align}
 \FDR(\SUD(\mbf{t}),m_0,F)&\leq  \frac{m_0}{m}\E_{\mathrm{DU}(m,m_0-1)} [q(\hat{u})]\nonumber\\
 &\leq  \frac{m_0}{m} \frac{\rho(u^+_\delta)}{u^+_\delta} +   \frac{m_0}{m} \P_{\mathrm{DU}(m,m_0-1)}[(\Omega_\delta(m_0-1))^c],\label{equ1}
 \end{align}
because $q$ is non-decreasing,  $u^+_\delta$ is positive and by \eqref{equ-borninfsup}. Next, by \eqref{ineg-conc}, 
we obtain the following upper-bound:
\begin{align}
 \FDR(\SUD(\mbf{t}),m_0,F) &\leq  \frac{m_0}{m} \frac{\rho(u^+_\delta)}{u^+_\delta} +   \frac{m_0}{m} 2\exp\left\{-2 m \left(\delta-\nu-\frac{1}{m}\right)_+^2 (1-\nu/\zeta)_+ \right\}.\label{upper-bound}
 \end{align}

\subsubsection{Lower bound}

In the model $\DU$ with $m_0<m$, we have $\hat{u}>0$ a.s. and thus 
\begin{align}
 \FDR(\SUD(\mbf{t}),m_0,F\equiv1)&= \frac{m_0}{m} \E_{\DU} \left(\frac{ \hat{\mathbb{G}}_{m_0}(\rho(\hat{u}))}{\hat{u}} \right)\nonumber\\
 &\geq\frac{m_0}{m} \E_{\DU} \left(\frac{ \hat{\mathbb{G}}_{m_0}(\rho(\hat{u}))}{\hat{u}}  \ind{\Omega_\delta(m_0)}\right)\nonumber\\
 &\geq \frac{m_0}{m} \E_{\DU} \left(\frac{ \hat{\mathbb{G}}_{m_0}(\rho(u^-_\delta))}{u^+_\delta}  \ind{\Omega_\delta(m_0)}\right)\nonumber\\
 &\geq \frac{m_0}{m} \frac{\rho(u^-_\delta)}{u^+_\delta}  - \frac{m_0}{m} \frac{1}{1-\zeta}\P_{\DU} \left[(\Omega_\delta(m_0))^c\right],\nonumber
   \end{align}
by \eqref{equ-borninfsup} and because $u^+_\delta\geq 1-\zeta$. From \eqref{ineg-conc}, we obtain the lower-bound
   \begin{align}
 \FDR(\SUD(\mbf{t}),m_0,F\equiv1)&\geq \frac{m_0}{m} \frac{\rho(u^-_\delta)}{u^+_\delta}  - \frac{m_0}{m} \frac{1}{1-\zeta}2\exp\left\{-2 m \left(\delta-\nu-\frac{1}{m}\right)_+^2 (1-\nu/\zeta)_+ \right\}.\label{lower-bound}
   \end{align}
 Finally, \eqref{upper-bound} and  \eqref{lower-bound} yield the result.  

\subsubsection{Proof for random mixture model}

In the $\RM$ model with $\pi_0=\zeta$, the distribution of $m_0$ is binomial with parameters $(m,\zeta)$. In particular, $\nu$ is random. However, we can write
\begin{align}
\E[ \FDP(\SUD(\mbf{t}),m_0)] \leq& \:\E\left[ \FDP(\SUD(\mbf{t}),m_0) \ind{\nu \leq \gamma} \right] \nonumber\\
&+  2\:\P\left(|m_0/m-\pi_0|>\gamma -1/m\right).\label{equ-interm}
  \end{align}
Additionally, using Hoeffding's \citeyearpar{Hoeff1963} inequality, we can write
\begin{align}\label{equ:Hoeff}\P\left(|m_0/m-\pi_0|>\gamma -1/m\right)\leq 2 e^{-2 m (\gamma -1/m)_+^2}.\end{align}
Combining \eqref{equ-interm} and \eqref{equ:Hoeff} with \eqref{equ-final} finishes the proof.

\subsection{Proof of Corollary~\ref{thmasymp}}\label{proof:thmasymp}

First consider the $\FM$ model with $m_0<m$ (the case $m_0=m$ is trivial). Let $\nu_m=\max_{k\in\{m_0-1,m_0\}}\{|k/m-\zeta_m|\}$ and consider $\delta_m \in (\nu_m, 1)$ that satisfies for large $m$,
\begin{equation}
2(1-\nu_m/\zeta_m) (\delta_m-\nu_m-1/m)^2 = (\log m)/m,
\end{equation}
so that $e^{-2 m \left(\delta_m-\nu_m-1/m\right)_+^2 (1-\nu_m/\zeta_m)_+ }=1/m$ for large $m$.
Since $\nu_m\leq 2/m$ by assumption, we have $\delta_m \propto  \sqrt{(\log m)/m}$.
From Theorem~\ref{thm}, it is sufficient to prove 
$$\frac{\rho(u^+_{\delta_m})-\rho(u^-_{\delta_m})}{u^+_{\delta_m}} =O(\delta_m/(1-\zeta_m)).$$
From Section~\ref{sec:finite}, this holds for the linear critical value function. This also holds for the AORC as soon as $\lambda_m/m < v_{\delta_m}$, which is the case for large $m$ by assumption.

The proof in the $\RM$ model is similar by taking additionally $\gamma_m \propto  \sqrt{(\log m)/m}$ such that 
$
2 e^{-2 m (\gamma_m -1/m)_+^2}=1/m.
$

\section*{Acknowledgements}

This work was supported by the French Agence Nationale de la Recherche (ANR grant references: ANR-09-JCJC-0027-01, ANR-PARCIMONIE, ANR-09-JCJC-0101-01) and the French ministry of foreign and european affairs (EGIDE - PROCOPE project number 21887 NJ).

\appendix

 \section{Formulas for FDP distribution} \label{fdp-formulas}

From Theorem~\ref{th-RV} and \eqref{SUD-comb}, we can compute the exact c.d.f. of the FDP of any SUD procedure in the following way, for each fixed number  $m\geq 2$ of hypotheses.

\begin{corollary}
Let $\lambda\in\{1,\dots,m\}$ and consider any threshold collection $\mbf{t}$. Fix an arbitrary $x\in(0,1)$. 
Then the following holds:
\begin{itemize}
\item[(i)] in the model $\RM$, for any $\pi_0\in[0,1]$,  $F\in\mathcal{F}$, we have
\begin{align}
\P( \FDP(\SUD(\mbf{t})) \leq x)=&\: \sum_{k=0}^{\lambda-1} \sum_{j=0}^{\lfloor x k \rfloor} \surm_{m,\pi_0,F}((t_\lambda\wedge t_j)_{1\leq j \leq m},k,j)  \nonumber\\
+&\sum_{k=\lambda}^m  \sum_{j=0}^{\lfloor x k \rfloor} \sdrm_{m,\pi_0,F}((t_\lambda\vee t_j)_{1\leq j \leq m},k,j) 
;\label{FDPsud_rm}
\end{align}
\item[(ii)] in the model $\FM$, for any $m_0\in\{0,\dots,m\}$,  $F\in\mathcal{F}$,  we have
\begin{align}
\P( \FDP(\SUD(\mbf{t})) \leq x)=&\: \sum_{k=0}^{\lambda-1} \sum_{j=0\vee (k-m+m_0)}^{m_0\wedge\lfloor x k \rfloor} \sufm_{m,m_0,F}((t_\lambda\wedge t_j)_{1\leq j \leq m},k,j)  \nonumber\\
+&\sum_{k=\lambda}^m  \sum_{j=0\vee (k-m+m_0)}^{m_0\wedge\lfloor x k \rfloor} \sdfm_{m,m_0,F}((t_\lambda\vee t_j)_{1\leq j \leq m},k,j) .\label{FDPsud_fm}
\end{align}
\end{itemize}
\end{corollary}

\section{FDR(SD) can exceed FDR(SU) in an extreme configuration}\label{result-extrem}

\begin{lemma}
Consider the $\FM$ model with $F(x)=\ind{x\geq 1}$ (i.e., all the $p$-values under the alternative are constantly equal to $1$). Consider the threshold collection $\mbf{t}$ defined by $t_k=t_0$, $1\leq k\leq m-1$ and $t_m=1$, for some $t_0\in(0,1)$. Then we have for any $\lambda\in\{1,\dots,m-1\}$,
\begin{align*}
\FDR(\mbox{SUD}_{\lambda}(\mbf{t}))&= 1-(1-t_0)^{m_0};\\
\FDR( \mbox{SU}(\mbf{t})) &= m_0/m.
\end{align*}
In particular, $\FDR(\mbox{SD}(\mbf{t}))>\FDR( \mbox{SU}(\mbf{t})) $ for $t_0>1-(1-m_0/m)^{1/m_0}$.
\end{lemma}
The proof is straightforward and is left to the reader. As an illustration, for $m=10$ and $m_0=7$, $1-(1-m_0/m)^{1/m_0}\simeq 0.158$.

\section{DU is an LFC for the $k$-FWER}\label{sec:FWER_LFC}

We state here for the sake of completeness a straightforward generalization of Lemma~1 of \citet{FG2009}
(see also Lemma~2.2 of \citealp{Gont2010}) concerning the LFCs of multiple testing procedures under
a class of type I criteria containing in particular the $k$-FWER
(but not the FDR, as pointed out in the introduction). This result should be considered as
already known by experts in the field, although we failed to locate a precise reference for it.
The setting considered assumes independence of $p$-values corresponding to true nulls, but 
is more general than the fixed mixture model, since 
the $p$-values corresponding to true null hypotheses are only assumed to be stochastically larger
than a uniform variable on [0,1]; also, the $p$-values corresponding to alternatives are not assumed to
be identically distributed nor independent.

\begin{lemma}
Let $m\geq 1$ and $m_0 \in \{0,\dots,m\}$ be fixed. Let $\bp=(p_1,\ldots,p_m)$ be a family $p$-values
with distribution by $P$ such that $(p_i)_{1\leq i \leq m_0}$ form an independent family 
of variables, each stochastically lower bounded by a uniform variable.
Assume that $\delta$ is a multiple testing procedure rejecting all hypotheses having $p$-value
less than a data-dependent threshold $t^*(\bp)$.
Let $R$ be a type I error criterion taking the form
\[
R(P,\delta) = \mbe_{\bp \sim P}[\phi(V_m(\delta(\bp)))],
\]
where $V_m$ is defined in \eqref{eq:defV} and $\phi$ is a function from $\mbn$ to $\mbr$.

Assume the two following conditions are satisfied:

\noindent (i) $t^*$ is a nonincreasing function of each $p$-value;\\
\noindent (ii) $\phi$ is nondecreasing.

Then it holds that
\[
R(P,\delta) \leq R(\DU,\delta),
\]
that is, $\DU$ is an LFC for $\delta$ among the set of distributions satisfying the properties
described above.
\end{lemma}
\begin{proof}
Using (i) and (ii) together entails that $\bp \mapsto \phi(V_m(\delta(\bp)))$ is a nonincreasing 
function of each $p$-value. Denote $\bp_0=(p_1,\ldots,p_{m_0},0,\ldots,0)$ the $p$-value family 
obtained by replacing $p_i$ by 0 for $i> m_0$, and $P_0$ the distribution of $\bp_0$
when $\bp$ has distribution $P$.
 Obviously we have
\[
\mbe_{\bp \sim P}[\phi(V_m(\delta(\bp)))] \leq \mbe_{\bp \sim P}[\phi(V_m(\delta(\bp_0)))]
= \mbe_{\bp \sim P_0}[\phi(V_m(\delta(\bp)))]
\]
Now applying Lemma A.11 as cited by \citet{Gont2010}, we obtain
\[
\mbe_{\bp \sim P_0}[\phi(V_m(\delta(\bp)))] \leq \mbe_{\bp \sim \DU}[\phi(V_m(\delta(\bp)))]\,,
\]
and thus the conclusion.
\end{proof}
A straightforward (though less immediately interpretable) extension
of this result 
to procedures that are not necessarily threshold-based is to replace assumption
(i) by (i'): $\bp \mapsto V_m(\delta(\bp))$ is a non increasing function of
each $p$-value.

\bibliography{biblio}
\bibliographystyle{abbrvnat}

%



%


\vskip .65cm
\noindent
Institut f\"{u}r Mathematik,   Universit\"{a}t Potsdam,  Germany
\vskip 2pt
\noindent
E-mail: gilles.blanchard@math.uni-potsdam.de
\vskip 2pt
\noindent
Institut f\"{u}r Mathematik,   Humboldt-Universit\"{a}t,  Germany
\vskip 2pt
\noindent
E-mail: dickhaus@math.hu-berlin.de
\vskip 2pt
\noindent
LPMA, UPMC, Universit\'e Paris 6, France
\vskip 2pt
\noindent
E-mail: etienne.roquain@upmc.fr
\vskip 2pt
\noindent
LPMA, UPMC, Universit\'e Paris 6, France
\vskip 2pt
\noindent
E-mail: fanny.villers@upmc.fr
\vskip .3cm
\end{document}